\pdfoutput=1
\documentclass[letterpaper,12pt]{article}

\usepackage{graphicx}
\usepackage{amsfonts}
\usepackage{amsmath}
\usepackage{amssymb}
\usepackage{amsthm}
\usepackage{mathrsfs}
\usepackage{enumerate}
\usepackage{color}
\usepackage{hyperref}
\hypersetup{colorlinks=true,linkcolor=blue,citecolor=blue,urlcolor=blue}
\usepackage{fancyhdr} 
\usepackage{natbib}
\usepackage{placeins} 

\DeclareRobustCommand{\rev}[1]{{#1}}

\newtheorem{proposition}{Proposition}
\newtheorem*{proposition*}{Proposition}
\newtheorem{theorem}{Theorem}
\newtheorem*{theorem*}{Theorem}

\newtheorem{obs}{Observation}
\newtheorem{corollary}{Corollary}
\newtheorem*{corollary*}{Corollary}
\newtheorem{lemma}{Lemma}
\newtheorem*{lemma*}{Lemma}
\newtheorem{claim}{Claim}

\theoremstyle{definition}
\newtheorem{definition}{Definition}
\newtheorem*{definition*}{Definition}

\newtheorem*{axiom*}{Axiom}

\newtheorem*{remark*}{Remark}
\newtheorem{example}{Example}
\newtheorem*{example*}{Example}
\newtheorem*{eff*}{Efficiency}
\newtheorem*{ir*}{Individual Rationality}
\newtheorem*{feas*}{Budget-Feasibility}
\newtheorem*{imp*}{Non-Imposition}
\newtheorem*{budget*}{Budget-Balancedness}

\newcommand{\1}{\boldsymbol{1}}

\newcommand{\argmax}{\mathop{\rm arg~max}\limits}

\begin{document}
\begin{titlepage}

\title{Information Requirements for Resolving Implementation Impossibilities
}

\author{
Yuta Nakamura\footnote{
Corresponding author.
School of Economics and Business Administration,
Yokohama City University,
22-2 Seto, Kanazawa-ku, Yokohama 236-0027, Japan.
E-mail: \texttt{y\_naka@yokohama-cu.ac.jp}.
}
}

\date{\today}
\maketitle
\thispagestyle{empty}

\begingroup
\setlength{\leftmargini}{2em}
\begin{abstract}
We develop a framework for assessing implementation impossibilities
through the information required to resolve them.
Domain covers describe information used to select a mechanism
without restricting admissible preferences.
Coarsening a cover weakens the required information.
In quasi-linear public decision problems with a status quo,
we study two impossibilities of dominant-strategy implementation:
combining \emph{efficiency} and \emph{individual rationality}
with either \emph{budget feasibility} or
\emph{budget balancedness}.
\rev{Under an interval cover based on bounds on maximal social surplus,
we characterize the social choice rules implementable in dominant
strategies and satisfying the first set of axioms.
Under a surplus-level cover based on exact surplus levels,
we obtain the corresponding characterization for the second set of axioms.}
Both covers are minimal: no strictly coarser cover permits
implementation under the respective requirements.
These economically interpretable minimal covers deepen our
understanding of the informational severity of each impossibility.
\end{abstract}
\endgroup

\medskip



\end{titlepage}

\section{Introduction}
Implementation impossibilities identify limits to achieving a given
social objective under incentive constraints over a specified preference
domain. They typically deliver a binary conclusion: no single mechanism
implements the desired rule, or satisfies the desired collection of
properties, on that domain. Yet this conclusion leaves open an economically
meaningful question: how much additional information about the realized
preferences would make implementation possible throughout the same domain?
Some impossibilities may be resolved with relatively coarse information,
while others may require much more precise information. This paper develops
a framework for assessing the severity of an impossibility in terms of the
information required to make implementation possible. It characterizes
tight informational requirements and uses their coarseness to compare
different impossibilities.

Two benchmarks help frame our question. Without additional information,
a single mechanism must accommodate the entire preference domain, and
the original impossibility applies. At the other extreme, complete
knowledge of the realized profile eliminates the elicitation problem:
any single-valued social choice rule can be implemented simply by
prescribing its outcome at that profile. \rev{The substantive question is which intermediate forms of information suffice
for implementation and how much they can be weakened.}

Consider the preference domain $\mathcal{D}^n$ determined by the economic
environment considered by the designer, which we call the \rev{\emph{ambient domain}}.
Suppose that no social choice rule satisfying a given set of desirable
properties can be implemented by a single mechanism on this domain.
A conventional restricted-domain result may show that a rule with these
properties can be implemented on a subset $D\subsetneq\mathcal{D}^n$.
This result alone says nothing about implementation outside $D$.
Now suppose that a rule satisfying the same properties can also be
implemented on another domain $C$, with $C\cup D=\mathcal{D}^n$, using
a different mechanism. Then, information identifying either $C$ or $D$
as containing the realized profile suffices to select a mechanism that
implements a rule satisfying those properties. The profile itself need
not be known, and every profile in the original domain is handled by
at least one of the two mechanisms. The cover $\{C,D\}$ thus identifies
information sufficient to resolve the original impossibility without
excluding any admissible profile. More generally, determining which
covers permit implementation and which do not allows us to assess an
impossibility in terms of the information needed to resolve it.

We formalize this idea through domain-contingent implementation.
Let $F$ denote a social choice rule on \rev{the ambient domain} $\mathcal{D}^n$.
The design variables are a domain cover
$\{D^\lambda\}_{\lambda\in\Lambda}$, satisfying
$\bigcup_{\lambda\in\Lambda}D^\lambda=\mathcal{D}^n$,
and a family of mechanisms $\{M^\lambda\}_{\lambda\in\Lambda}$,
both indexed by the same set $\Lambda$.
For a realized preference profile $\succsim\in\mathcal{D}^n$,
a label \rev{$\lambda\in\Lambda$} is valid if $\succsim\in D^\lambda$.
The cover describes an informational requirement: at every realized
profile, the information available at the implementation stage must
suffice to identify at least one valid label.
Informally, we say that $F$ is
$\{D^\lambda\}_{\lambda\in\Lambda}$-implementable if every valid label \rev{$\lambda$}
is associated with a mechanism $M^\lambda$ whose equilibrium outcomes
are prescribed by $F$.
Thus, for any realized preference profile, information identifying a
valid label enables an outcome prescribed by $F$ to be realized in
equilibrium by selecting and announcing the corresponding mechanism.
\rev{When the cover is the singleton $\{\mathcal{D}^n\}$,
$\{\mathcal{D}^n\}$-implementation coincides with conventional
implementation on $\mathcal{D}^n$.}
Figure~\ref{fig1} summarizes this sequence: the designer commits to
the cover and mechanisms. After the preference profile is realized,
a valid label is specified and the corresponding mechanism is announced.
Individuals then interact through that mechanism.

\begin{figure}[tbp]
    \centering
    \includegraphics[
        width=\linewidth,
        trim={106bp 192bp 130bp 191bp},
        clip
    ]{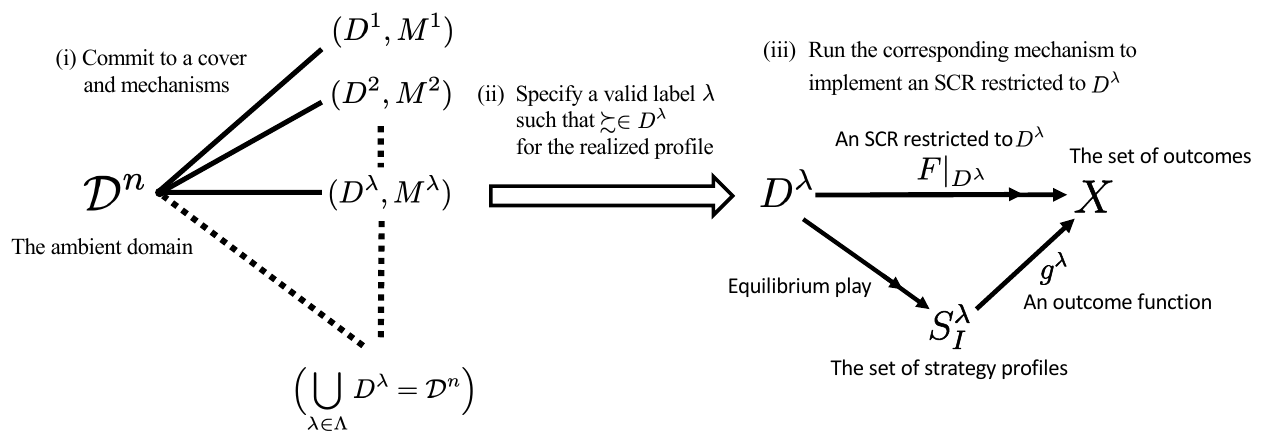}
    \caption{Domain-contingent implementation.}
    \label{fig1}
\end{figure}

A conventional restricted-domain result establishes implementability
on a particular subset of the ambient domain, but by itself leaves
implementation outside that subset unaddressed.
Our approach identifies information sufficient to make implementation
possible throughout the ambient domain, without excluding any admissible
preference profiles or relaxing the original social and incentive
requirements. By organizing suitable restricted domains into a cover
with associated mechanisms, the designer can achieve the social objective
at every admissible profile when a valid label is assigned.
The cover also allows us to examine how much information is needed to
overcome the original impossibility and how far this informational
requirement can be weakened.

To compare the informational requirements of different covers, we use
coarseness. A valid label $\lambda$ reveals that the realized preference
profile lies in $D^\lambda$. The larger $D^\lambda$ is, in the sense
of set inclusion, the more profiles remain compatible with that label
and hence the less precise the information it conveys.
Formally, $\{D^\lambda\}_{\lambda\in\Lambda}$ is coarser than
$\{C^\kappa\}_{\kappa\in K}$ if, for every $\kappa\in K$,
there exists $\lambda(\kappa)\in\Lambda$ such that
$C^\kappa\subseteq D^{\lambda(\kappa)}$.
It is strictly coarser if at least one of these inclusions is proper.
Information identifying a valid label $\kappa$ for \rev{$\{C^\kappa\}_{\kappa\in K}$}
then suffices to identify a valid label $\lambda(\kappa)$ for
\rev{$\{D^\lambda\}_{\lambda\in\Lambda}$} without additional information.
Coarseness thus compares how far the set of possible preference profiles
must be narrowed down before the corresponding mechanism is used \rev{in our implementation process}.

The informational severity of two impossibilities can then be compared.
Consider two sets of requirements on the same ambient domain and under
the same solution concept, for which implementation is possible under
covers $\{D^\lambda\}_{\lambda\in\Lambda}$ and
$\{C^\kappa\}_{\kappa\in K}$, respectively.
Suppose further that no cover strictly coarser than
$\{D^\lambda\}_{\lambda\in\Lambda}$ permits implementation for the
first problem, and no cover strictly coarser than
$\{C^\kappa\}_{\kappa\in K}$ does so for the second.
If $\{D^\lambda\}_{\lambda\in\Lambda}$ is strictly coarser than
$\{C^\kappa\}_{\kappa\in K}$, the information described by
$\{D^\lambda\}_{\lambda\in\Lambda}$ suffices for the first problem
but cannot suffice for the second.
This provides an ordinal comparison: the second impossibility is more
severe because information sufficient to resolve the first remains
insufficient to resolve the second.

The covers also allow us to assess the severity of each impossibility
in terms of the information needed for implementation.
For each problem, the corresponding cover describes the content and
precision of information sufficient to make implementation possible
throughout the ambient domain.
The absence of a strictly coarser cover permitting implementation under
the same requirements identifies a limit on weakening this information.
The cover therefore captures the severity of the impossibility in
informational terms. This assessment allows us both to compare the
severity of multiple impossibilities and to evaluate a single
impossibility on its own.

\rev{In addition, this} approach is especially useful when such covers admit a simple
economic interpretation. When two covers are described by information
about the same simple statistic, such as maximal social surplus,
the difference in their informational requirements can be expressed
in common terms. Resolving an impossibility through such a cover
therefore helps us understand its severity in economic terms.
We illustrate this approach in public decision problems using two
comparable covers defined in terms of surplus information.

\subsection{Main results}

We apply the framework to public decision problems with quasi-linear
preferences and a status quo.
The ambient domain $\mathcal{D}^n$, denoted by $V^n$
in this application, contains all valuation profiles
with the status quo normalized to zero and valuations
of the other alternatives allowed to take any real values.
Throughout this application, implementation is evaluated in dominant
strategies. On this ambient domain, the following impossibilities hold
(\citealt{GreenLaffont1977}; \citealt{Jackson2003}, Section~3.9):

\begin{itemize}
    \item[(i)]
    There exists no social choice rule that is dominant-strategy
    implementable, \emph{efficient}, \emph{individually rational},
    and \emph{budget feasible}.

    \item[(ii)]
    There exists no social choice rule that is dominant-strategy
    implementable, \emph{efficient}, \emph{individually rational},
    and \emph{budget balanced}.
\end{itemize}

\rev{Since \emph{budget feasibility} is a weaker requirement than
\emph{budget balancedness}, statement (i) implies statement (ii).}
We investigate which covers permit implementation under each set of
requirements. Both covers considered below describe information about
maximal social surplus, defined as the largest total valuation gain
from the status quo.

To overcome (i), we consider an interval cover indexed by positive
labels \rev{$r\in\mathbb{R}_{++}$}. The domain associated with $r$ consists of profiles whose
maximal social surplus is either zero or lies in
$\left[\frac{n-1}{n}r,r\right]$, where $n\geq2$ is the number of
individuals. These domains cover the entire ambient domain.
At a positive-surplus profile, identifying a valid label therefore
requires bounds on maximal surplus, while allowing its exact value
to remain unknown.

Theorem~\ref{thm1} shows that a social choice rule is
\emph{efficient}, \emph{individually rational}, \emph{budget feasible},
and implementable under this cover if and only if it is a
quasi-equal-division rule.
These rules select an alternative that maximizes social surplus
and equalize individuals' final utilities at a nonnegative level,
potentially resulting in a net outflow of transfers.
Moreover, no strictly coarser cover permits implementation of a rule
satisfying the same requirements, even if the rule and the mechanisms
are changed. The interval cover therefore identifies a limit on
weakening information sufficient to overcome (i).

This result has a simple interpretation in disputes involving a small
number of parties. Consider two neighboring landowners choosing among
boundary settlements, with the unresolved dispute as the status quo.
Each owner has a monetary valuation of each settlement, and transfers
can compensate an owner for accepting a settlement that creates a
positive combined gain.
With two owners, when the maximal combined gain is positive, a valid
label need only place it between $r/2$ and $r$.
The corresponding mechanism elicits their valuations, selects an
efficient settlement, and determines transfers that leave neither
owner worse off without an outside subsidy.
More generally, for $n$ parties and a fixed $r$, the interval has
width $r/n$. Fewer participants therefore allow less precise surplus
information under this cover.

To overcome (ii), we consider a surplus-level cover, also indexed by
$r>0$. The domain associated with $r$ consists of profiles whose
maximal social surplus is either zero or exactly $r$.
Thus, at a positive-surplus profile, a valid label identifies the exact
surplus level.

\rev{Theorem~\ref{thm2} shows that a social choice rule satisfies
\emph{efficiency}, \emph{individual rationality}, and \emph{budget balancedness}
and is implementable under this cover if and only if it is an
equal-division rule.}
These rules select an alternative that maximizes social surplus and
distribute the entire surplus equally among individuals, so transfers
sum to zero. Again, no strictly coarser cover permits implementation
under the same requirements.
In particular, enlarging its domains to pool distinct positive surplus
levels destroys budget-balanced implementation.
\rev{\emph{Budget balancedness}} therefore remains demanding even when information
about the realized preferences can be used at the implementation stage.

The two results make the informational difference between (i) and (ii)
explicit. The interval cover is strictly coarser than the surplus-level
cover, and the minimality of the latter implies that the former cannot
permit implementation under (ii).
The contrast expresses the difference in terms of the precision of
information about the same economic statistic: bounds on positive
maximal surplus suffice under \emph{budget feasibility}, whereas the
cover resolving the budget-balance impossibility requires its exact
level. \rev{Moreover, each} minimality result is also informative separately, because
it identifies a limit on weakening the surplus information sufficient
to resolve the corresponding impossibility.

\subsection{Related literature}

The paper builds on the literature examining how preference domains
shape social choice and implementation
\citep[e.g.,][]{Holmstrom1979,Moulin1980}.
Our framework organizes restricted domains into a cover of a fixed
ambient domain and interprets information identifying a valid label
as a requirement for selecting a mechanism.
This connects implementability on restricted domains to implementation
throughout the original environment.
Comparing covers then allows us to examine how far the informational
requirement can be weakened while preserving the social and incentive
requirements.

The paper also relates to mechanism design with verifiable information.
Studies of implementation with evidence include
\citet{GreenLaffont1986},
\citet{BullWatson2007},
\citet{Ben-PorathLipman2012},
\citet{KartikTercieux2012},
\citet{Peralta2019},
and \citet{BanerjeeChenSun2024}.
\rev{Other work on mechanism design with verifiable information} examines partial verification, communication,
commitment, and disclosure
\citep{DeneckereSeverinov2008,SherVohra2015,
Ben-PorathDekelLipman2019,Pram2021}.
These studies examine how evidence about agents' types or the state
affects incentives and attainable outcomes.
In the hard-evidence formulation, presenting evidence, such as a
document or record, rules out types or states at which that evidence
is unavailable.
Our framework specifies an informational requirement through a domain
cover and asks whether information identifying a valid label permits
implementation. \rev{Our focus is on characterizing the information sufficient for
mechanism selection and the limits on weakening that information
while preserving implementation.}

Other studies examine costly or imperfect verification
\rev{of information privately held by agents}
\citep{Townsend1979,GaleHellwig1985,BorderSobel1987,
MookherjeePng1989,Ben-PorathDekelLipman2014,PereyraSilva2023}.
Information available after allocation also plays a role in contracting
and mechanism design, including public ex post information
\citep{RiordanSappington1988}, subsequent elicitation of realized
utilities \citep{Mezzetti2004}, and ex post verification with limited
penalties \citep{MylovanovZapechelnyuk2017}.
These contributions study how the availability, cost, or timing of
information affects feasible contracts and allocations.
Our contribution is to assess and compare the severity of implementation
impossibilities in terms of the information needed to resolve them.

\rev{The rest of this paper is organized as follows.}
Section~\ref{quasi} introduces the framework and its formal analysis
in the public decision setting.
Section~\ref{public} presents the characterization and minimality
result under \emph{budget feasibility}.
Section~\ref{public-sec4} gives the corresponding results under
\rev{\emph{budget balancedness}}.
Section~\ref{conc} concludes.
Proofs omitted from the main text are collected in
Appendices~\ref{app:revelation}--\ref{app:verification}.

\section{The model}\label{quasi}
\subsection{Preliminaries}
Let $I=\{1,2, \dots, n\}$ ($n \geq 2$) be the finite set of \emph{individuals}, and $A=\{\phi, a_1, \dots, a_m\}$ ($|A| \geq 2$) the finite set of \emph{alternatives}, where $\phi$ denotes the \emph{status quo}.
An \emph{outcome} is a vector $(a, p) \in X \equiv A \times \mathbb{R}^n$, where $p \in \mathbb{R}^n$ denotes monetary transfers (received amount) among the individuals.
A \emph{valuation function} is a function $v_i: A \to \mathbb{R}$, where
for simplicity, we normalize each $v_i: A \to \mathbb{R}$ as $v_i(\phi)=0$.
Let $V$ be the set of all valuation functions.
Given $v_i \in V$, individual $i$'s \emph{quasi-linear utility} is $u_i(a, p;v_i) = v_i(a)+p_i$.
A \emph{valuation profile} is an $n$-tuple of valuation functions $v=(v_1, \dots, v_n) \in  V^n$.
The ambient domain in this environment is $V^n$. This is the valuation representation of the general preference domain \(\mathcal D^n\) \rev{introduced in the introduction}: each normalized valuation function represents a quasi-linear preference over outcomes. We keep this ambient domain fixed throughout the analysis.
For each $v \in  V^n$ and each $I' \subset I$, we denote $v_{I'} \equiv (v_j)_{j \in I'}$ and $v_{-I'} \equiv (v_j)_{j \in I \setminus I'}$, respectively.
For each valuation profile \(v\in V^n\), define the maximal social surplus by
$$
S(v) \equiv \max_{b\in A}\sum_{j=1}^n\bigl(v_j(b)-v_j(\phi)\bigr)
=\max_{b\in A}\sum_{j=1}^n v_j(b),
$$
where \rev{the last equality} follows from the normalization \(v_j(\phi)=0\) for every \(j\in I\). Since \(\phi\in A\), we have \(S(v)\geq 0\) for every \(v\in V^n\). Let $$V_0 \equiv \{v\in V^n:S(v)=0\}$$
denote the set of all zero-surplus profiles.

A \emph{social choice rule} is a correspondence $F: V^n \to 2^{A \times \mathbb{R}^n} \setminus \{\emptyset\}$ that associates each valuation profile $v \in V^n$ with a nonempty set of outcomes $F(v) \subset  A \times \mathbb{R}^n$.
We consider the following four axioms, each of which imposes a requirement on the outcomes prescribed by $F$ at every profile in the ambient domain.

First, {\it efficiency} requires that an outcome chosen by a social choice rule maximize the social surplus.
 \begin{eff*}
For each $v \in V^n$ and each $(a, p) \in F(v)$,
$a \in \argmax_{b \in A} \left\{ \sum_{j=1}^n v_j(b) \right\}$.
\end{eff*}

\medskip

Second, {\it individual rationality} requires that for any outcome chosen by a social choice rule, each individual's final utility must be weakly greater than that at the status quo.
\begin{ir*}
For each $i \in I$, each $v \in V^n$, and each $(a, p) \in F(v)$, $u_i\big(a, p; \, v_i \big) \geq 0$.
\end{ir*}

\medskip

Third, {\it budget feasibility} requires that for any outcome chosen by a social choice rule, net transfers among individuals must be non-positive.
Recall that each $p_i$ denotes the transfer amount received by individual $i$.

\begin{feas*}
For each $v \in V^n$ and each $(a, p) \in F(v)$, $\sum_{j=1}^n p_j \leq 0$.
\end{feas*}

\medskip

Finally, {\it budget balancedness} requires that for any outcome chosen by a social choice rule, net transfers among individuals must be zero.
Clearly, {\it budget balancedness} is a stronger requirement than {\it budget feasibility}.

\begin{budget*}
For each $v \in V^n$ and each $(a, p) \in F(v)$, $\sum_{j=1}^n p_j = 0$.
\end{budget*}

For each $i \in I$, let $S_i$ be a strategy space of individual $i$, and let $S_I \equiv S_1 \times \dots \times S_n$.
A \emph{strategy profile} is an $n$-tuple of strategies $s=(s_1, \dots, s_n) \in  S_I$.
An \emph{outcome function} is a function $g: S_I \to A \times \mathbb{R}^n$  that maps each strategy profile $s \in S_I$ to an outcome $g(s)  \in  A \times \mathbb{R}^n$.
For notational convenience, we often write $g(s)=(a(s), p(s))$.
A \emph{mechanism} is a pair $M=(S_I, g)$ of strategy spaces and an outcome function.

Given a mechanism $M$ and a valuation function $v_i \in V$,  a strategy $s_i \in S_i$ is a \emph{dominant strategy for $i$ at} $(M, v_i)$ if for each $s_{-i} \in \prod_{j \neq i}S_j$ and each $s'_{i} \in S_i$, 
\begin{align*}
u_i\left(g(s_i, s_{-i}); v_i\right) \geq u_i\left(g(s'_i, s_{-i}); v_i\right). 
\end{align*}
Let $DS_i(M, v_i)$ be the set of dominant strategies for $i$ at $(M, v_i)$.
A strategy profile $s \in S_I$ is a \emph{dominant strategy equilibrium at $(M, v)$} if for each $i \in I$, $s_i \in DS_i(M, v_i)$ holds.
Let $DE(M, v) \equiv \prod_{j=1}^n DS_j(M, v_j)$ be the set of dominant strategy equilibria at $(M, v)$.
The set of \emph{dominant strategy outcomes at $(M, v)$} is denoted by 
\[g\left(DE(M, v)\right) \equiv \left\{g(s) \in A \times \mathbb{R}^n: s \in DE(M, v) \right\}.\]

The following definition summarizes the standard notion of dominant strategy implementation in its full form.

\begin{definition}\label{DSI}
A mechanism $M$ \emph{implements a social choice rule $F$ in dominant strategies} if for each $v \in V^n$,
$g\left(DE(M, v)\right) = F(v)$.
A social choice rule $F$ is \emph{dominant strategy implementable} if there exists $M$ that implements $F$ in dominant strategies.

\end{definition}

As noted in the introduction, our purpose is to overcome the following two impossibility results due to \citet{GreenLaffont1977}: (i) there exists no social choice rule that is efficient, individually rational, budget feasible, and dominant strategy implementable; and (ii) there exists no social choice rule that is efficient, individually rational, budget balanced, and dominant strategy implementable. Since budget balancedness is stronger than budget feasibility, statement (i) immediately implies statement (ii).

\subsection{Domain-contingent implementation}\label{MDG}
We study the information requirements for resolving the impossibilities in Section 2.1 while keeping \(V^n\) as the ambient domain. We represent these requirements by a domain cover.
A \rev{\emph{domain cover}} is an indexed family \(\{D^\lambda\}_{\lambda\in\Lambda}\) of nonempty subsets of \(V^n\) such that $\bigcup_{\lambda\in\Lambda}D^\lambda=V^n.$
Each \(\lambda\in\Lambda\) is a \rev{\emph{label}}, and \(D^\lambda\) is its \rev{\emph{support}}.\footnote{Distinct labels may have identical supports. This distinction matters for implementing multivalued social choice rules, as illustrated by Observation 1.} A label \(\lambda\) is valid at a valuation profile \(v\in V^n\) if \(v\in D^\lambda\). \rev{The supports in \(\{D^\lambda\}_{\lambda\in\Lambda}\)} may overlap, so a profile may admit several valid labels.
A valid label \rev{\(\lambda\)} indicates that the realized profile belongs to its support. A larger support \rev{\(D^\lambda\)} leaves more profiles compatible with the label \rev{\(\lambda\)} and therefore conveys less information about the realized profile. The covering condition ensures that at least one valid label exists at every profile in the ambient domain.

\begin{example}
The following are examples of domain covers.
\begin{itemize}
\item \rev{The trivial cover \(\{V^n\}\) has a single label, which is valid at every profile; we omit the label index from the notation.} Supplying a valid label therefore requires no information about the realized profile.
\item The fully revealing cover \(\{\{v\}\}_{v\in V^n}\) has one label for each profile. \rev{For each realized profile \(v\), the unique valid label is \(v\) itself.}
\item Let 
$D^L \equiv \{v\in V^n:S(v)\leq100\}$
and  $D^H \equiv \{v\in V^n:S(v)>100\}$.
Under the cover \(\{D^L,D^H\}\), a valid label \rev{\(\lambda\in\{L,H\}\)} reveals whether maximal social surplus exceeds \(100\).
\end{itemize}
\end{example}

The domain cover \(\{D^\lambda\}_{\lambda\in\Lambda}\) and a corresponding family of mechanisms
$\{M^\lambda\}_{\lambda\in\Lambda}$ with
$M^\lambda=(S_I^\lambda,g^\lambda)$
are fixed in advance. Once a valuation profile is realized, a label $\lambda\in\Lambda$
is supplied by the person or institution responsible for applying
the mechanism in practice.\footnote{
We call this person or institution the \emph{executor}.
The executor's role and the consequences of invalid labels
are discussed in Section~\ref{Discussion}.}
The corresponding mechanism $M^\lambda$ is then announced. Individuals then play that mechanism. Implementation requires the prescribed outcomes to be obtained whenever the supplied label is valid, as formalized below.
 
\medskip

\begin{definition}\label{def-D}
A family of mechanisms $\left\{M^{\lambda} \right\}_{\lambda \in \Lambda}$ is said to \emph{$\{D^{\lambda}\}_{\lambda \in \Lambda}$-implement a social choice rule $F$} if the following two conditions are satisfied:
\begin{itemize}
\item[(i)] for each  $\lambda \in \Lambda$ and each $v \in D^{\lambda}$, $DE(M^{\lambda}, v) \neq \emptyset$ and $g^{\lambda} \left(DE(M^{\lambda}, v)\right) \subset F(v)$,
\item[(ii)] For each $v \in V^n$, $\bigcup_{\left\{\lambda \in \Lambda: v \in D^{\lambda}\right\}} g^{\lambda} \left(DE(M^{\lambda}, v)\right) = F(v)$.
\end{itemize}
A social choice rule $F$ is \emph{$\{D^{\lambda}\}_{\lambda \in \Lambda}$-implementable} if there exists $\left\{ M^{\lambda} \right\}_{\lambda \in \Lambda}$ that $\{D^{\lambda}\}_{\lambda \in \Lambda}$-implements $F$.
\end{definition}

Condition (i) requires every valid label to yield a nonempty set of dominant-strategy outcomes contained in \(F(v)\).\rev{\footnote{\rev{Since each $D^{\lambda} \in \{D^{\lambda}\}_{\lambda \in \Lambda}$ may not be a Cartesian product of subsets of $V$, we cannot directly define \emph{strategy-proofness} of social choice rules restricted to $D^{\lambda}$.
Therefore, we consider indirect mechanisms to explicitly define individuals' strategy spaces.
}}} Condition (ii) requires every outcome in \(F(v)\) to arise under at least one valid label. Thus, a single valid label need not generate all of \(F(v)\), but the outcome sets across valid labels must collectively exhaust \(F(v)\).
\rev{Whenever the supplied label $\lambda$ is valid at $v$, the mechanism $M^\lambda$ must have dominant-strategy equilibria, and every resulting outcome must belong to $F(v)$.}

The following observations provide two benchmarks.

\begin{obs}\label{D-exist}
Every social choice rule $F$ is implementable under a domain cover whose supports are singletons. In particular, every single-valued social choice rule is \(\{\{v\}\}_{v\in V^n}\)-implementable.
\end{obs}

\begin{proof}
Let $\Lambda_0 \equiv \{(v,x)\in V^n\times X:x\in F(v)\}$.
For each \((v,x)\in\Lambda_0\), set \(D^{(v,x)} \equiv \{v\}\) and choose a mechanism \(M^{(v,x)}\) with a singleton strategy space for each individual and constant outcome \(x\). Its dominant-strategy outcome set at \(v\) is \(\{x\}\). Taking the union over the labels valid at \(v\) yields \(F(v)\), so both conditions of Definition \ref{def-D} hold. If \(F\) is single valued, one label per profile suffices. 
\end{proof}

\begin{obs}\label{D-equiv} 
The following statements hold:
\begin{itemize}
\item[(i)] There exists no social choice rule that is efficient, individually rational, budget feasible, and $\{V^n\}$-implementable,
\item[(ii)] There exists no social choice rule that is efficient, individually rational, budget balanced, and $\{V^n\}$-implementable.
\end{itemize}
\end{obs}

\begin{proof}
Under the trivial cover \(\{V^n\}\), there is a single label, valid at every profile. Definition 2 therefore reduces to
\rev{$g(DE(M,v))=F(v)$ for every $v\in V^n$,}
which is conventional dominant-strategy implementation. Both claims follow from the impossibility results in Section 2.1. 
\end{proof}

Observation 1 shows that every social choice rule can be implemented when the label reveals the entire valuation profile. Observation 2 shows that implementing rules with the desired properties is impossible under the trivial cover $\{V^n\}$, which requires no information about the realized profile. These benchmarks motivate the search for minimal information requirements for implementation. We compare such requirements through the coarseness of domain covers.

Let \(\{C^\kappa\}_{\kappa\in K}\) and \(\{D^\lambda\}_{\lambda\in\Lambda}\) be domain covers. We compare covers through enlargement of their supports. We say that \(\{D^\lambda\}_{\lambda\in\Lambda}\) is strictly coarser than \(\{C^\kappa\}_{\kappa\in K}\), written \(\{C^\kappa\}_{\kappa\in K}\prec\{D^\lambda\}_{\lambda\in\Lambda}\) if for each $\kappa \in K$, there exists $\lambda(\kappa) \in \Lambda$ such that $C^{\kappa} \subset D^{\lambda(\kappa)}$, and $C^{\kappa'} \subsetneq D^{\lambda(\kappa')}$ holds for at least one $\kappa' \in K$. If a label \(\kappa\) is valid at \(v\) under \(\{C^\kappa\}_{\kappa\in K}\), then the corresponding label \(\lambda(\kappa)\) is also valid under \(\{D^\lambda\}_{\lambda\in\Lambda}\), since \(v\in C^\kappa\subseteq D^{\lambda(\kappa)}\). Thus, any information sufficient to identify a valid label under the former cover also suffices under the latter. Enlarging a support leaves more profiles compatible with the corresponding label, so identifying that label requires less precise information about the realized profile. The comparison therefore captures a weakening of informational requirements through enlargement of supports.

For a given collection of social-choice axioms, we consider the class of domain covers under which at least one social choice rule satisfying those axioms is implementable. We seek covers that are minimal in the following sense.\footnote{We allow a domain cover that is not a partition of $V^n$, but rather covers it. This is because there may be cases where a non-partitioned domain cover requires minimal information about individual valuations within a class of domain covers. In fact, we will provide examples of non-partitioned domain covers that require minimal information about individual valuations.}

\begin{definition}
A domain cover $\{D^{\lambda}\}_{\lambda \in \Lambda}$ \emph{requires minimal information within a class of domain covers} if it belongs to that class and there exists no domain cover $\{C^{\kappa}\}_{\kappa \in K}$ in the class such that $\{D^{\lambda}\}_{\lambda \in \Lambda} \prec \{C^{\kappa}\}_{\kappa \in K}$.\footnote{Our definition of strict coarseness can classify a cover with properly nested supports as strictly coarser than itself. For example, a cover \(\{V_0,V^n\}\) is strictly coarser than itself under this definition. However, \(\{V^n\}\) is also strictly coarser than \(\{V_0,V^n\}\), so the redundant cover is excluded by minimality. \rev{More generally, any cover containing two supports $D^\lambda\subsetneq D^{\lambda'}$ is strictly coarser than itself under our definition and therefore cannot be minimal.}}
\end{definition}

A minimal cover identifies an informational requirement that permits implementation and cannot be further weakened under the coarsening relation defined above. Although minimal covers need not be unique, finding one with a simple economic interpretation, such as a cover defined by maximal social surplus, can deepen our understanding of the relevant impossibility. When minimal covers associated with different impossibilities are comparable, their coarseness relation also provides an ordinal comparison of their informational demands.

We also define when label validity can be verified from dominant-strategy play. A mechanism \(M^\lambda\) \emph{verifies the validity of label \(\lambda\)} if for each \(v\in D^\lambda\) and each \(v'\in V^n\setminus D^\lambda\), the sets
$DE(M^\lambda,v)$ and 
$DE(M^\lambda,v')$
are nonempty and disjoint. A family of mechanisms \emph{verifies label validity} if this condition holds for all labels.
This condition allows the information represented by the label to be checked ex post. Provided that individuals play dominant strategies, their observed strategy profile reveals whether the realized valuation profile lies in \(D^\lambda\): the same strategy profile cannot arise both when the label is valid and when it is invalid. The consequences of an invalid label are discussed in Section~\ref{Discussion}.

\begin{definition}
A social choice rule \(F\) is \emph{verifiably \(\{D^\lambda\}_{\lambda\in\Lambda}\)-implementable} if there exists a family of mechanisms that \(\{D^\lambda\}_{\lambda\in\Lambda}\)-implements \(F\) and verifies label validity.
\end{definition}

\section{Efficiency, individual rationality,  and budget feasibility}\label{public}
In this section, we study the information requirements for implementing social choice rules that satisfy \emph{efficiency}, \emph{individual rationality}, and \emph{budget feasibility}.
Observation~\ref{D-equiv}(i) rules out such implementation under the trivial cover $\{V^n\}$.
We consider a domain cover defined by intervals of maximal social surplus.

For each $r\in\mathbb{R}_{++}$, let
\[
D^r \equiv V_0 \cup
\left\{
v\in V^n:
\frac{n-1}{n}r \leq S(v)\leq r
\right\}.
\]
Recall that $n$ is the number of individuals and $V_0$ is the set of valuation profiles whose maximal social surplus is zero.
Clearly, $\{D^{r}\}_{r \in \mathbb{R}_{++}}$ is a domain cover, i.e., $\bigcup_{r \in \mathbb{R}_{++}} D^{r}=V^n$, and we call this family the interval cover.
Given a positive real number $r \in \mathbb{R}_{++}$, $D^r$ is the set of valuation profiles whose maximal social surplus is zero or lies within the closed interval $\big[\frac{n-1}{n}r, r\big]$.
Thus, a label \rev{$r\in\mathbb{R}_{++}$} is valid at $v$ precisely when
$S(v)=0$ or $\frac{n-1}{n}r\leq S(v)\leq r$,
as illustrated in Figure~\ref{fig2}.
A valid label provides bounds on positive maximal social surplus without ruling out zero surplus.
\begin{figure}[]
 \begin{center}
  \includegraphics[width=85mm, viewport={155 470 652 564}]{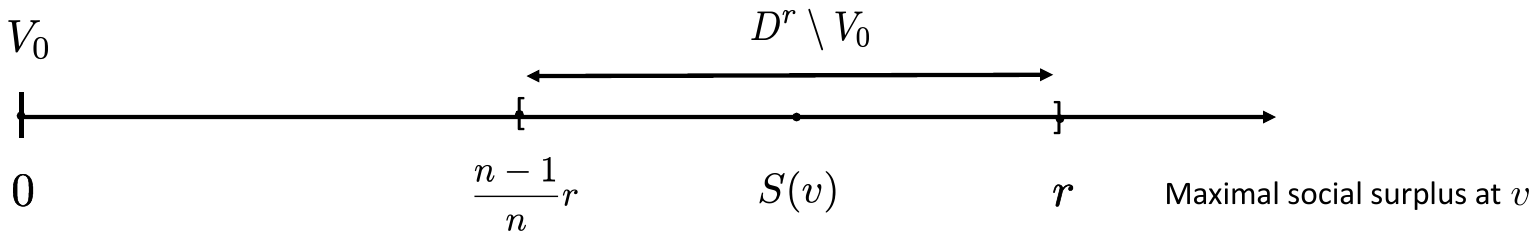}
 \end{center}
 \caption{Surplus value $S(v)$ compatible with a valid label $r$}
 \label{fig2}
\end{figure}
The informational requirement is expressed in terms of the single aggregate $S(v)$.
Supplying a valid label does not require identifying the full valuation profile or an alternative that attains the maximal surplus.
For fixed $r$, the positive-surplus interval has width $r/n$, which increases as $n$ decreases; when $n=2$, it is $[r/2,r]$.

\subsection{Quasi-equal-division rule}
We now define the class of quasi-equal-division rules.

\begin{definition}\label{SR}
A social choice rule $F^*$ is a \emph{quasi-equal-division rule} if for each $v\in V^n$, the following five conditions hold:
\begin{itemize}
\item[(i)] for each $(a,p)\in F^*(v)$,
$a \in \argmax_{b \in A} \, \sum_{j=1}^n v_j(b)$,
\item[(ii)] for each $(a,p)\in F^*(v)$, if $S(v)=0$, then $a=\phi$,
\item[(iii)] for each $(a,p)\in F^*(v)$,
$u_i(a,p;v_i)=u_j(a,p;v_j)$ for all $i,j\in I$,
\item[(iv)] for each $(a,p)\in F^*(v)$,
$0\leq\sum_{j=1}^n u_j(a,p;v_j)\leq S(v)$,
\item[(v)] for each $c\in[0,S(v)]$, there exists $(a,p)\in F^*(v)$ such that
$\sum_{j=1}^n u_j(a,p;v_j)=c$.
\end{itemize}
\end{definition}

Condition (i) states that a quasi-equal-division rule selects alternatives that maximize the social surplus. Condition (ii) specifies the way to break ties, particularly when the maximal social surplus is zero. 
Condition (iii) ensures that payments are arranged so that individuals' final utilities are equal.
Condition (iv) says that the total amount of the final utilities must lie within the closed interval between $0$ and $S(v)$.
Condition (v) requires every total utility level in $[0,S(v)]$ to be attained by some outcome in $F^*(v)$.\footnote{Condition (iv) is imposed to ensure \emph{budget feasibility}. This is because if \rev{$\sum_{j=1}^n u_j\left(a, p; v_j\right) \leq S(v)=\max_{b \in A} \sum_{j=1}^n v_j(b)$}, then condition (i) implies that 
\rev{\[\max_{b \in A} \sum_{j=1}^n v_j(b)+\sum_{j=1}^n p_j =\sum_{j=1}^n u_j\left(a, p; v_j\right)\leq \max_{b \in A} \sum_{j=1}^n v_j(b), \]}
which implies $\sum_{j=1}^n p_j \leq 0$. 
Condition (v) is necessary to guarantee the implementability of $F^*$ under $\{D^{r}\}_{r \in \mathbb{R}_{++}}$. 
Indeed, each valuation profile $v^* \in V^n \setminus V_0$ must lie in multiple valid supports $D^r \in \{D^{r}\}_{r \in \mathbb{R}_{++}}$, and distinct mechanisms for each such $D^r$ must be used to guarantee implementability of $F^*$.
This inevitably forces $F^*(v^*)$ to include multiple outcomes.
For further details, see the explanation of the intuition behind the proof in Section \ref{Intuition1}. 
}
Figure \ref{fig3} illustrates the possible final utilities of an individual under these rules for a given profile $v \in V^n$.
Here, $u_i\left(F^*(v); v_i\right)$ denotes the set of $i$'s final utilities achievable under $F^*(v)$. Conditions (iii)-(v) imply that this set coincides with the interval $\left[0, \frac{1}{n} S(v)\right]$.
\begin{figure}[!htbp]
 \begin{center}
  \includegraphics[width=130mm, viewport={20 300 489 532}]{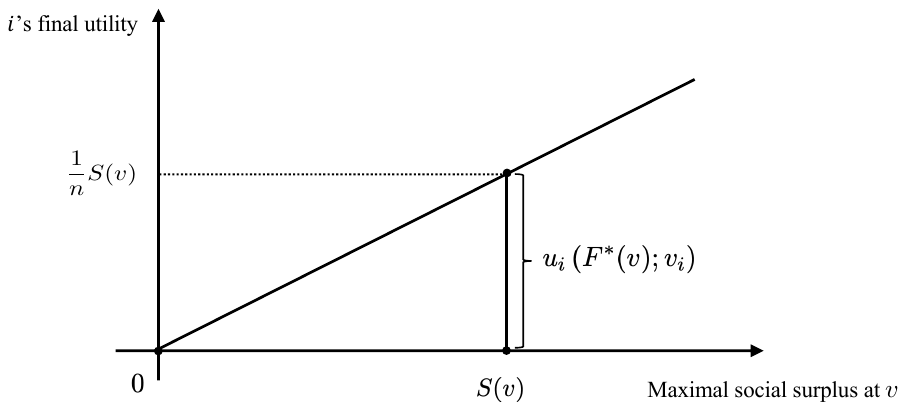}
 \end{center}
 \caption{\rev{Individual $i$'s possible final utilities under a quasi-equal-division rule}}
 \label{fig3}
\end{figure}
\FloatBarrier

\subsection{An implementing family of mechanisms}\label{Sec_mechanisms}
We introduce a family of revelation mechanisms that
$\{D^r\}_{r\in\mathbb{R}_{++}}$-implements a quasi-equal-division rule.
Each individual reports a valuation function in $V$, \rev{which can be identified with $\mathbb{R}^{|A|-1}$,} so the strategy space is finite-dimensional.\footnote{For discussions of the role of strategy spaces in implementation, see \citet{DuttaSenVohra1995} and \citet{SaijoTatamitaniYamato1996}.}

Let $\1_{\phi}:A\to\{0,1\}$ be the indicator of the status quo:
$\1_{\phi}(\phi)=1$ and $\1_{\phi}(a)=0$ for all $a\neq\phi$.
For each $r\in\mathbb{R}_{++}$, define a revelation mechanism
$M^r=(V^n,g^r)$ with outcome function
$g^r=(a^r,p_1^r,\dots,p_n^r)$ satisfying, for each 
$v\in V^n$ and each $i\in I$,
\begin{equation*}\label{mech}
\begin{aligned}
a^r(v)
&\in \argmax_{b\in A}
\left\{
\sum_{j=1}^n v_j(b)
+\frac{n-1}{n}r\,\1_{\phi}(b)
\right\},\\
p_i^r(v)
&=
\sum_{j\neq i}v_j\bigl(a^r(v)\bigr)
+\frac{n-1}{n}r\,\1_{\phi}\bigl(a^r(v)\bigr)
-\frac{n-1}{n}r,
\end{aligned}
\end{equation*}
where ties involving $\phi$ are broken in favor of a non-status-quo maximizer.
Any remaining ties are resolved according to a rule fixed in advance.
Thus, $a^r(v)$ is an efficient non-status-quo alternative when
$S(v)\geq\frac{n-1}{n}r$, and equals $\phi$ otherwise.
For each individual $i$, the transfer is zero if $a^r(v)=\phi$ and equals
$\sum_{j\neq i}v_j(a^r(v))-\frac{n-1}{n}r$ otherwise.

This class of mechanisms belongs to the class of \emph{affine maximizers} introduced by \citet{Roberts1979}, and thus under the revelation mechanism $(V^n, g^r)$, reporting the true valuation function is the unique dominant strategy for each individual of each type $v_i \in V$,
as shown in the proof of Proposition~\ref{prop1}.
Consequently, $DE(M^r,v)=\{v\}$ \rev{for all
$r\in\mathbb{R}_{++}$ and $v\in V^n$},
whether or not $r$ is valid at $v$.
Thus, $M^r$ verifies the validity of label $r$:
profiles inside and outside $D^r$ generate distinct dominant-strategy report profiles.

\begin{proposition}\label{prop1}
The family of revelation mechanisms
$\{M^r\}_{r\in\mathbb{R}_{++}}$
verifiably $\{D^{r}\}_{r \in \mathbb{R}_{++}}$-implements a quasi-equal-division rule $F^*$. 
\end{proposition}

Proposition~\ref{prop1} guarantees that, for every $v\in V^n$
and every label $r$ valid at $v$, the corresponding mechanism
$M^r$ has a unique dominant-strategy outcome $g^r(v)\in F^*(v)$.
This outcome is efficient, individually rational, and budget feasible.
Therefore, every individual is weakly better off than under the
status quo, although \rev{budget balancedness} may not hold.

\subsection{Characterization}
Theorem~\ref{thm1} characterizes the efficient, individually rational, and budget-feasible social choice rules implementable under $\{D^r\}_{r\in\mathbb{R}_{++}}$.
It also establishes the cover's minimality among domain covers permitting implementation of a rule satisfying these three properties.
\begin{theorem} \label{thm1}
The following statements hold:
\begin{itemize}
\item[(i)] a social choice rule $F$ is efficient, individually rational, budget feasible,
and $\{D^{r}\}_{r \in \mathbb{R}_{++}}$-implementable if and only if $F$ is a quasi-equal-division rule,
\item[(ii)] For any domain cover
$\{D^\lambda\}_{\lambda\in\Lambda}$ such that
$\{D^r\}_{r\in\mathbb{R}_{++}}
\prec
\{D^\lambda\}_{\lambda\in\Lambda}$,
no social choice rule is efficient, individually rational,
budget feasible, and
$\{D^\lambda\}_{\lambda\in\Lambda}$-implementable.
\end{itemize}
\end{theorem}

The mechanisms constructed to prove the ``if'' direction of
Theorem~\ref{thm1}(i) also verify label validity,
yielding the following corollary.

\begin{corollary}\label{coro-p1}
Any quasi-equal-division rule is verifiably $\{D^{r}\}_{r \in \mathbb{R}_{++}}$-implementable.
\end{corollary}

Theorem~\ref{thm1} and Corollary~\ref{coro-p1} show that the interval cover permits verifiable implementation of every quasi-equal-division rule.
No strictly coarser cover permits implementation of any rule satisfying efficiency, individual rationality, and budget feasibility.
The cover therefore identifies a minimal informational requirement for resolving the impossibility in Observation~\ref{D-equiv}(i), expressed in terms of bounds on maximal social surplus.

\subsection{Intuition behind the proof of Theorem~\ref{thm1}}\label{Intuition1}
We first explain the ``only if'' direction of statement~(i).
Consider any implementation of an efficient, individually rational,
and budget feasible rule $F$ under the interval cover.
Fix a label $r\in\mathbb{R}_{++}$.
A version of the revelation principle adapted to our setting
(see Lemma~\ref{lem_rev0} in Appendix~\ref{app:revelation})
allows us to work with a direct-revelation mechanism $(V^n,g^r)$
in which truthful reporting is a dominant strategy and
$g^r(v)\in F(v)$ whenever $r$ is valid.
The proof shows that, at positive-surplus profiles, each individual's
final utility takes the familiar VCG form:\footnote{See
Lemma~\ref{lem-imple} in Appendix~\ref{app:interval} for the formal argument.}
\begin{equation}
u_i\bigl(g^r(v);v_i\bigr)
=
S(v)+h_i(v_{-i})
\qquad\text{for all }v\in D^r\setminus V_0,
\label{VCG}
\end{equation}
where $h_i$ depends only on the other individuals' valuations.

Holding $v_{-i}$ fixed keeps $h_i(v_{-i})$ constant.
Thus, a one-unit increase in maximal social surplus raises
individual $i$'s final utility by one unit.
In Figure~\ref{fig-proof}, this relationship is represented by
a $45^\circ$ line, whose vertical position is determined by
$h_i(v_{-i})$.
\begin{figure}[]
 \begin{center}
  \includegraphics[width=130mm, viewport={20 275 483 507}]{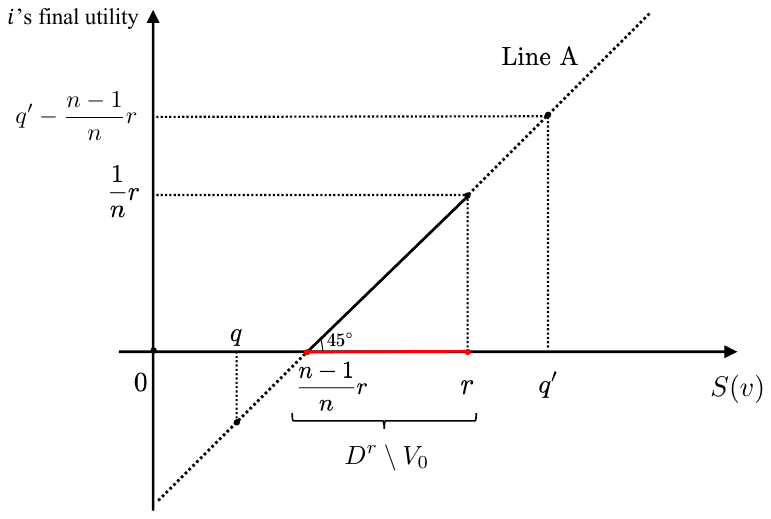}
 \end{center}
\caption{Intuition behind the proof of Theorem~\ref{thm1}}
 \label{fig-proof}
\end{figure}
First, consider the lower endpoint, where
$S(v)=\frac{n-1}{n}r$.
Individual rationality requires each individual's final utility
to be nonnegative.
The $45^\circ$ line therefore cannot lie below Line A,
which passes through $\bigl(\frac{n-1}{n}r,0\bigr)$.

Second, consider the upper endpoint, where $S(v)=r$.
Since the $45^\circ$ line cannot lie below Line A,
each individual's final utility must be at least $\frac{1}{n}r$.
Budget feasibility, however, requires the sum of all individuals'
final utilities to be no greater than \rev{$S(v)=r$}.\footnote{\label{footfeas}
\emph{Efficiency} and \emph{budget feasibility} imply that
\begin{equation}
\sum_{j=1}^n u_j\bigl(g^r(v);v_j\bigr)
=
S(v)+\sum_{j=1}^n p_j^r(v)
\leq S(v),
\label{feas-foot}
\end{equation}
so the sum of the final utilities cannot exceed  the maximal social surplus. 
}
These two requirements leave no room for anyone to receive
more than $\frac{1}{n}r$: everyone must receive exactly $\frac{1}{n}r$.
This pins down the $45^\circ$ line as Line A.
Thus, throughout the positive-surplus part of $D^r$,
\[
u_i\bigl(g^r(v);v_i\bigr)
=
S(v)-\frac{n-1}{n}r.
\]
The same argument applies to every dominant-strategy outcome
under a valid label.

Now fix a valuation profile $v$ with $S(v)>0$ and vary the valid label.
A label $r$ is valid if and only if
$\frac{n-1}{n}r\leq S(v)\leq r$,
so the valid labels are precisely
\[
r\in\left[S(v),\frac{n}{n-1}S(v)\right].
\]
At $r=S(v)$, each individual's final utility is $S(v)/n$.
As $r$ increases to $\frac{n}{n-1}S(v)$,
everyone's final utility falls continuously to zero.
Thus, varying the valid label generates every common utility
level between zero and $S(v)/n$ (see Figure~\ref{fig-proof}).
By Definition~\ref{def-D}, the outcomes across valid labels
exhaust $F(v)$, so this is exactly the utility range under $F(v)$.
The zero-surplus case is handled separately in the proof
and yields the status quo with zero transfers.
Together, these observations show that $F$ is a quasi-equal-division rule.

The same intuition explains statement~(ii).
A strictly coarser cover must contain a support $D^\lambda$
that strictly includes some $D^r$.
Suppose that such a cover permits implementation of a rule
satisfying the three properties.
The original support $D^r$ pins down Line A, and incentive
compatibility requires the same utility formula at any additional
positive-surplus profile.

Since $D^r$ already contains all zero-surplus profiles,
an added profile must have surplus either below
$\frac{n-1}{n}r$ but above zero, or above $r$.
If its surplus is $0<q<\frac{n-1}{n}r$,
Line A assigns everyone a negative utility,
violating \emph{individual rationality}.
If instead its surplus is $q'>r$,
Line A assigns everyone a utility of
\[
q'-\frac{n-1}{n}r>\frac{1}{n}q'.
\]
Total utility then exceeds the available surplus $q'$,
violating \emph{budget feasibility}.\footnote{See Footnote \ref{footfeas} once again. If the maximal social surplus is $q'$ and each individual's final utility exceeds $\frac{1}{n}q'$, then equation \eqref{feas-foot} does not hold, contradicting \emph{budget feasibility}.} 
Thus, extending the support in either direction violates
one of these two requirements.
No strictly coarser cover can therefore permit implementation
of an efficient, individually rational, and budget feasible rule.

\subsection{Utilities and transfers across valid labels}\label{Sec_utilities}
We compare the dominant-strategy outcomes of the mechanisms
$\{M^r\}_{r\in\mathbb{R}_{++}}$ in Section~\ref{Sec_mechanisms}
across labels valid at the same valuation profile.
If $S(v)=0$, every label is valid and every mechanism selects
the status quo with zero transfers.
We therefore focus on a fixed profile $v$ with $S(v)>0$.

If a label $r$ is valid at $v$, then
$\frac{n-1}{n}r\leq S(v)\leq r$, or equivalently,
\begin{equation}
S(v)\leq r\leq\frac{n}{n-1}S(v).
\label{choice-r}
\end{equation}
Under every such label, the outcome $g^r(v)$ satisfies
\emph{efficiency} (see Section~\ref{Sec_mechanisms}),
and each individual's final utility is
\[
u_i\bigl(g^r(v);v_i\bigr)
=
S(v)-\frac{n-1}{n}r.
\]
Thus, at a fixed profile, every individual obtains a higher
final utility under a smaller valid label $r$.

The aggregate net outflow from the individuals is\footnote{
Using \emph{efficiency} and the utility formula above, we obtain
\begin{align*}
-\sum_{i=1}^n p_i^r(v)
&=
\sum_{i=1}^n v_i\bigl(a^r(v)\bigr)
-\sum_{i=1}^n u_i\bigl(g^r(v);v_i\bigr)\\
&=
S(v)-n\left(S(v)-\frac{n-1}{n}r\right)\\
&=
(n-1)\bigl(r-S(v)\bigr).
\end{align*}
}
\[
-\sum_{i=1}^n p_i^r(v)
=
(n-1)\bigl(r-S(v)\bigr).
\]
A smaller valid label therefore reduces this outflow and
leaves more of the surplus with the individuals,
while maintaining \emph{efficiency}.
At the smallest valid label, $r=S(v)$,
each individual obtains utility $S(v)/n$,
and \emph{budget balancedness} holds, i.e., $\sum_{i=1}^n p_i^r(v)=0$.
At the largest valid label, $r=\frac{n}{n-1}S(v)$,
each individual obtains zero utility,
and the aggregate net outflow equals the entire surplus $S(v)$.
Within this family, \emph{budget balancedness} at a positive-surplus
profile therefore requires $r=S(v)$, whereas a label $r$ can be valid
under the interval cover without satisfying this equality.
Section~\ref{public-sec4} studies the information requirements
when \emph{budget balancedness} must hold under every valid label.

\section{Efficiency, individual rationality, and budget balancedness}\label{public-sec4}
In this section, we study the information requirements for
implementing social choice rules that satisfy \emph{efficiency},
\emph{individual rationality}, and \emph{budget balancedness}.
Observation~\ref{D-equiv}(ii) rules out such implementation
under the trivial cover $\{V^n\}$.
We consider a domain cover defined by exact levels of
maximal social surplus.
For each $r\in\mathbb{R}_{++}$, let
\[
C^r\equiv V_0\cup\{v\in V^n:S(v)=r\}.
\]
The family $\{C^r\}_{r\in\mathbb{R}_{++}}$ is a domain cover,
which we call the surplus-level cover.
A label $r$ is valid at $v$, i.e., $v \in C^r$, precisely when $S(v)=0$ or $S(v)=r$.
Thus, at a positive-surplus profile, the unique valid label
is $r=S(v)$, whereas every label is valid at a zero-surplus profile.
\rev{For every $r\in\mathbb{R}_{++}$, we have $C^r\subsetneq D^r$,
where $D^r$ is the support with label $r$ in the interval cover
introduced in Section~\ref{public}.}
Hence,
$\{C^r\}_{r\in\mathbb{R}_{++}}
\prec
\{D^r\}_{r\in\mathbb{R}_{++}}$.
\rev{The surplus-level cover therefore requires more precise information
about positive maximal social surplus than the interval cover.}

\subsection{Equal-division rules and an implementing family of mechanisms}
\label{sec4-mech}
We now define the class of equal-division rules.

\begin{definition}
A social choice rule $F^{**}$ is an \emph{equal-division rule}
if for each $v\in V^n$ and each $(a,p)\in F^{**}(v)$,
the following three conditions hold:
\begin{itemize}
\item[(i)]
$a\in\argmax_{b\in A}\sum_{j=1}^n v_j(b)$,
\item[(ii)]
if $S(v)=0$, then $a=\phi$,
\item[(iii)]
$u_i(a,p;v_i)=S(v)/n$ for all $i\in I$.
\end{itemize}
\end{definition}
Condition (i) imposes \emph{efficiency}, and condition (ii)
selects the status quo when $S(v)=0$.
Condition (iii) divides the entire maximal social surplus
equally among individuals.
Since $S(v)\geq0$, this ensures \emph{individual rationality};
together with condition (i), it also ensures
\emph{budget balancedness}.
Thus, an equal-division rule assigns the highest common utility
level permitted by a quasi-equal-division rule
(see Section~\ref{Sec_utilities}).

We use the same family $\{M^r\}_{r\in\mathbb{R}_{++}}$
introduced in Section~\ref{Sec_mechanisms}.
Truthful reporting remains the unique dominant strategy
for each individual, whether or not $r$ is valid under
the surplus-level cover.
Consequently, $M^r$ also verifies label validity:
profiles inside and outside $C^r$ generate distinct
dominant-strategy report profiles.

\begin{proposition}\label{prop2}
The family of revelation mechanisms
$\{M^r\}_{r\in\mathbb{R}_{++}}$
verifiably $\{C^r\}_{r\in\mathbb{R}_{++}}$-implements
an equal-division rule $F^{**}$.
\end{proposition}

Proposition~\ref{prop2} ensures that, once a valid label $r$
is supplied under the surplus-level cover,
simply announcing the corresponding mechanism $(V^n,g^r)$
suffices to implement an outcome prescribed by
the equal-division rule $F^{**}$.
\rev{However, supplying such a label can be demanding:}
when $S(v)>0$, $r$ must equal $S(v)$ exactly.

\subsection{Characterization}
For $|A|\geq3$, Theorem~\ref{thm2} characterizes the social choice
rules satisfying \emph{efficiency}, \emph{individual rationality},
and \emph{budget balancedness} that are implementable under
the surplus-level cover $\{C^{r}\}_{r \in \mathbb{R}_{++}}$.
It also establishes the cover's minimality among domain covers
permitting implementation of a rule satisfying these three properties.

\begin{theorem} \label{thm2}
Suppose that $|A| \geq 3$.
\begin{itemize}
\item[(i)] A social choice rule $F$ is efficient, individually rational, budget balanced, \\
and $\{C^{r}\}_{r \in \mathbb{R}_{++}}$-implementable if and only if $F$ is an equal-division rule,
\item[(ii)]For any domain cover $\{D^\lambda\}_{\lambda\in\Lambda}$ such that
$\{C^r\}_{r\in\mathbb{R}_{++}}
\prec
\{D^\lambda\}_{\lambda\in\Lambda}$, no social choice rule is efficient, individually rational, budget balanced, and $\{D^{\lambda}\}_{\lambda \in \Lambda}$-implementable.
\end{itemize}
\end{theorem}

The mechanisms constructed to prove the ``if'' direction of
Theorem~\ref{thm2}(i) also verify label validity.
The construction remains valid when $|A|=2$,
so the following corollary holds for all $|A|\geq2$.

\begin{corollary}\label{coro-pp4}
Any equal-division rule is verifiably
$\{C^r\}_{r\in\mathbb{R}_{++}}$-implementable.
\end{corollary}

\medskip

The intuition behind the proof of Theorem~\ref{thm2}
parallels that of Theorem~\ref{thm1},
so we only explain the reasoning for statement~(ii).
Suppose that a strictly coarser cover permits implementation
of a rule satisfying \emph{efficiency},
\emph{individual rationality}, and \emph{budget balancedness}.
Then some support $D^\lambda$ properly contains $C^r$
for some $r\in\mathbb{R}_{++}$.
By an argument similar to that used in the proof of
Theorem~\ref{thm1}, one can show that, at every
$v\in D^\lambda\setminus V_0$,
each individual's final utility is given by
\begin{equation}
S(v)-\frac{n-1}{n}r.
\label{int-thm2}
\end{equation}
Line A in Figure~\ref{fig-proof2} depicts this utility
as a function of maximal social surplus:
when $S(v)=q$, each individual obtains
$q-\frac{n-1}{n}r$, while at $S(v)=r$,
each obtains $\frac{1}{n}r$.
\begin{figure}[]
 \begin{center}
  \includegraphics[width=133mm, viewport={17 270 483 507}]{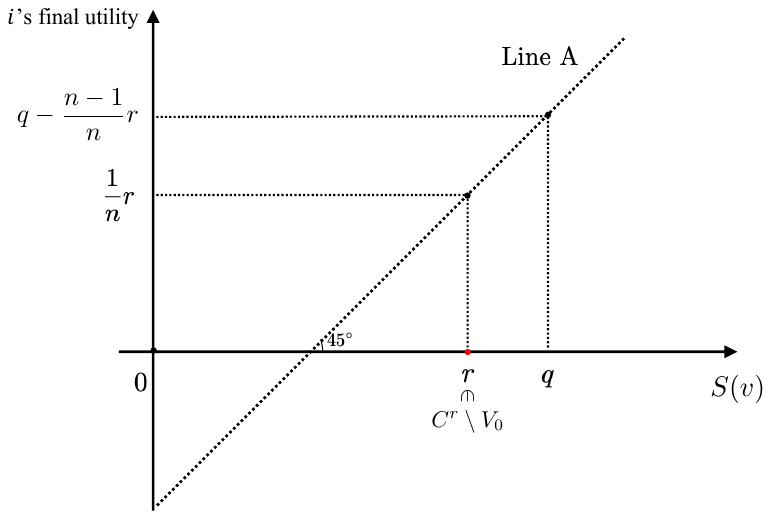}
 \end{center}
\caption{Sketch of the proof of Theorem~\ref{thm2}}
 \label{fig-proof2}
\end{figure}

Since these utilities are equal across individuals,
\emph{efficiency} and \emph{budget balancedness}
require each individual's final utility to equal $S(v)/n$
for all $v\in D^\lambda\setminus V_0$.
Since $C^r$ already contains all zero-surplus profiles
and all profiles with $S(v)=r$,
every additional profile in $D^\lambda\setminus C^r$
has $S(v)=q$ for some $q>0$ with $q\neq r$.
At such a profile, Line A gives each individual utility
$q-\frac{n-1}{n}r$, which differs from the required $q/n$.
Thus, no strictly coarser cover can permit implementation
of a rule satisfying all three properties.

Theorems~\ref{thm1} and \ref{thm2}
identify comparable minimal covers for the two impossibilities
in Observation~\ref{D-equiv}.
The interval cover permits bounds on positive surplus under
\emph{budget feasibility}, whereas the surplus-level cover
uses its exact value under \emph{budget balancedness}.
Because the former is strictly coarser than the latter,
these economically interpretable covers provide \rev{a comparison}
of the informational requirements for resolving
the two impossibilities.

\section{Conclusion}\label{Discussion}\label{conc}

This paper studies information requirements for resolving
implementation impossibilities through domain covers.
In public decision problems with a status quo,
the interval cover permits implementation of quasi-equal-division
rules satisfying \emph{efficiency},
\emph{individual rationality}, and \emph{budget feasibility},
while the surplus-level cover permits implementation of
equal-division rules satisfying \emph{efficiency},
\emph{individual rationality}, and \emph{budget balancedness}.
Both covers are minimal for their respective
requirements, and the interval cover is strictly coarser.
These economically interpretable covers clarify what information
can resolve each impossibility and provide \rev{a comparison}
of informational demands.

In applying the framework in practice, we distinguish
the designer's role from that of the \emph{executor},
the person or institution responsible for supplying a label
and applying the corresponding mechanism.
The designer fixes the domain cover and mechanism family
in advance.
Once a valuation profile $v$ is realized, the executor uses
available information to supply a label and announces
the corresponding mechanism, which individuals then play.
If the label is valid at $v$, every dominant-strategy outcome
belongs to $F(v)$.
\rev{The roles of the designer and the executor} may be performed by different parties or
by the same party at different stages.

For example, a mediator may serve as the executor in
a boundary dispute between two farmers.
If the farmers' combined valuations reflect the net
agricultural gains from resolving the dispute, then
land records and yield information may help establish bounds
on the maximal social surplus.
Sufficiently precise bounds could support a valid label under
the interval cover without identifying each farmer's valuation
function or the surplus-maximizing  boundary (efficient alternative).\footnote{
A similar interpretation applies to task allocation,
where the status quo involves no task being undertaken
and valuations reflect rewards net of task costs.}

The executor may nevertheless supply an invalid label.
For the family $\{M^r\}_{r\in\mathbb{R}_{++}}$
in Section~\ref{Sec_mechanisms}, truthful reporting remains
the unique dominant strategy throughout $V^n$,
regardless of label validity.
Provided that individuals play dominant strategies,
their observed reports therefore allow validity to be checked.

Recall that the same family $\{M^r\}_{r\in\mathbb{R}_{++}}$
is used under both the interval and surplus-level covers.
For the mechanism $M^r$, the consequences of an invalid label
depend on the realized surplus.
When $S(v)=0$, every label is valid under both covers,
and the outcome is the status quo with zero transfers.
If $0<S(v)<\frac{n-1}{n}r$, the label is invalid under
both covers, and the mechanism selects the status quo
with zero transfers: \emph{efficiency} fails,
while \emph{budget balancedness} holds.
If $\frac{n-1}{n}r\leq S(v)<r$, the label is valid under
the interval cover but invalid under the surplus-level cover.
The outcome satisfies \emph{efficiency} and
\emph{budget feasibility}, but not \emph{budget balancedness}.
Finally, if $S(v)>r$, the label is invalid under both covers.
The outcome satisfies \emph{efficiency}, but positive aggregate
transfers violate \emph{budget feasibility} and
\emph{budget balancedness}.
In all cases, \emph{individual rationality} is preserved,
so no individual is worse off than under the status quo.
Since a budget deficit arises precisely when $S(v)>r$,
an executor concerned with avoiding a deficit has a reason
to favor a larger label \rev{to reduce the risk of choosing $r<S(v)$}
when the surplus is uncertain.

A smaller label, however, leaves more surplus with
the individuals.
As shown in Section~\ref{Sec_utilities}, among valid labels
under the interval cover, a smaller $r$ gives every
individual a higher final utility.
When the surplus is uncertain, this creates a trade-off
between choosing a label large enough to avoid a deficit
and preserving individuals' utilities.
An excessively large label may even lead to the selection
of the status quo despite positive surplus,
violating \emph{efficiency}.

The present analysis takes label provision as given.
A natural extension would model the executor's objectives
alongside the costs and reliability of information acquisition.
This would allow us to study incentives to acquire information
and supply valid labels.

\clearpage
\appendix
\numberwithin{equation}{section}
\begin{center}
\Large{\bf Appendix: Omitted proofs}
\end{center}

\smallskip

\section{Proof preliminaries: an analogy to the revelation principle}
\label{app:revelation}
Before proceeding to the proofs, we provide an analogy to the \emph{revelation principle} (e.g., \citealt{DasguptaHammondMaskin1979}; \citealt{Myerson1979}) for $\{D^{\lambda}\}_{\lambda \in \Lambda}$-implementation.
For each $D^{\lambda} \subset V^n$, let $V_i(D^{\lambda})$ be the set of valuation functions possible for individual $i$ on $D^{\lambda}$, that is,
$V_i(D^{\lambda}) \equiv \left\{ v_i \in V: \exists v_{-i} \in V^{n-1}, \, (v_i, v_{-i}) \in D^{\lambda}\right\}$.
For instance, we have $V_i(D^{r})=V$ for each $r \in \mathbb{R}_{++}$.

Lemma~\ref{lem_rev0} states that if a family of mechanisms
$\{D^\lambda\}_{\lambda\in\Lambda}$-implements a social choice rule $F$,
then one can construct corresponding direct-revelation mechanisms
with strategy spaces $\prod_{j=1}^n V_j(D^\lambda)$.
Truthful reporting is a dominant strategy, and the truthful outcome
belongs to $F(v)$ whenever the label is valid.
Moreover, condition (i) allows the construction to preserve any
specified dominant-strategy outcome of the original mechanism
at a given profile.

\begin{lemma}\label{lem_rev0}
Suppose that there exists a family of mechanisms $\left\{ \left(S_I^{\lambda},g^{\lambda} \right) \right\}_{\lambda \in \Lambda}$ that $\{D^{\lambda}\}_{\lambda \in \Lambda}$-implements a social choice rule $F$. 
Fix any $\lambda \in \Lambda$, $\hat v \in D^{\lambda}$, and $\hat x \in g^{\lambda} \left(DE(\left(S_I^{\lambda},g^{\lambda} \right), \hat v)\right)$.
Then, there exists a mechanism $\left(\prod_{j =1}^n V_j (D^{\lambda}), \hat g^{\lambda} \right)$ that satisfies the following three conditions:
\begin{itemize}
\item[(i)] $\hat g^{\lambda}(\hat v)=\hat x$,
\item[(ii)] for each $v \in D^{\lambda}$, $\hat g^{\lambda}(v) \in F(v)$, 
\item[(iii)] for each $i \in I$, each $v \in \prod_{j =1}^n V_j (D^{\lambda})$, and each $v'_i \in V_i(D^{\lambda})$,
\[ u_i\left(\hat g^{\lambda}(v_i, v_{-i}); v_i\right) \geq u_i\left(\hat g^{\lambda}(v'_i, v_{-i}); v_i\right).\] 
\end{itemize}
\end{lemma}

\begin{proof}
By the definition of $\{D^{\lambda}\}_{\lambda \in \Lambda}$-implementation, for each $i \in I$ and each $v_i \in V_i (D^{\lambda})$, 
there exists a dominant strategy $s_i(v_i)$ at $\left( \left(S_I^{\lambda},g^{\lambda} \right), v_i\right)$.
In addition, $s_i(\cdot)$'s can be constructed so that $g^{\lambda} \left( \bigl(s_i(\hat v_i)\bigr)_{i \in I} \right)=\hat x$ holds particularly when $v=\hat v$.
Let $\hat{g}^{\lambda}:\prod_{j =1}^n V_j (D^{\lambda}) \to X$ be such that $\hat{g}^{\lambda}(v) = g^{\lambda}\left(\bigl( \, s_i(v_i)\bigr)_{i \in I} \right)$ for all $v \in \prod_{j =1}^n V_j (D^{\lambda})$.
Then, conditions (i) and (ii) are clearly satisfied by the definitions of $s_i(\cdot)$'s and $\{D^{\lambda}\}_{\lambda \in \Lambda}$-implementation.
Let us show that $\left(\prod_{j =1}^n V_j (D^{\lambda}), \hat{g}^{\lambda} \right)$ satisfies condition (iii).
Take any $i \in I$, any $v \in \prod_{j =1}^n V_j (D^{\lambda})$, and any $v'_i \in V_i(D^{\lambda})$.
Since $s_i(v_i)$ is a dominant strategy at $\left((S_I^{\lambda},g^{\lambda}),v_i\right)$,
\begin{align*}
u_i\left(\hat{g}^{\lambda}(v_i, v_{-i}); v_i\right)&=u_i\left(g^{\lambda}\bigl(s_i(v_i), (s_{j}(v_{j}))_{j \neq i}\bigr); v_i\right) \\
&\geq u_i\left(g^{\lambda}\bigl(s_i(v'_i), (s_{j}(v_{j}))_{j \neq i}\bigr); v_i\right)=u_i\left(\hat{g}^{\lambda}(v'_i, v_{-i}); v_i\right).
\end{align*}
Therefore, condition (iii) is also satisfied.
\end{proof}

\section{Proof of Proposition~\ref{prop1} and Theorem~\ref{thm1}}
\label{app:interval}
Recall that \(V_0=\{v\in V^n:S(v)=0\}\) is the set of all zero-surplus profiles. We begin with the following lemma.

\begin{lemma}\label{lem-imple}
Consider a domain cover $\{D^\lambda\}_{\lambda\in\Lambda}$.
Suppose that there exists a family of mechanisms $\left\{ \left(S_I^{\lambda},g^{\lambda} \right) \right\}_{\lambda \in \Lambda}$ that $\{D^{\lambda}\}_{\lambda \in \Lambda}$-implements an efficient, individually rational, and budget feasible social choice rule $F$. 
Fix any $r\in\mathbb{R}_{++}$ and $\lambda\in\Lambda$
such that $D^r\subseteq D^\lambda$.
Let $\left(V^n, \hat g^{\lambda} \right)$ be a mechanism that satisfies conditions (ii) and (iii) in Lemma \ref{lem_rev0},
 where
we write $\hat g^{\lambda}(v)=(\hat a^{\lambda}(v), \hat p^{\lambda}(v))$ for all $v \in V^n$.\footnote{Note that $ D^{r} \subseteq D^{\lambda}$ implies that $V_i (D^{\lambda})=V$ for all $i \in I$, and thus each individual's strategy space can be written as $V$.}
Then,
\[ \hat p^{\lambda}_i(v)= \sum_{j \neq i} v_j\big(\hat a^{\lambda}(v)\big) -\frac{n-1}{n}r  \ \ \ \ \  \text{for all $i \in I$, $v \in D^{\lambda} \setminus V_0$}. \]
\end{lemma}

\begin{proof}
{\bf Step 1:} Let us show that for each $i \in I$, each $v \in D^{\lambda}$, and each $v'_i \in V$ with $(v'_i, v_{-i}) \in D^{\lambda}$,
if $\hat a^{\lambda}(v_i, v_{-i}  )=\hat a^{\lambda}(v'_i, v_{-i} )$, then $\hat p^{\lambda}_i(v_i, v_{-i}  )=\hat p^{\lambda}_i(v'_i, v_{-i} )$.
Suppose, by contradiction, that there exist $i \in I$, $v \in D^{\lambda}$ and $v'_i \in V$ with $(v'_i, v_{-i}) \in D^{\lambda}$ such that $\hat a^{\lambda}(v_i, v_{-i}  )=\hat a^{\lambda}(v'_i, v_{-i} )$ and $\hat p^{\lambda}_i(v_i, v_{-i}  ) \neq \hat p^{\lambda}_i(v'_i, v_{-i} )$.
Without loss of generality, we consider the case with $\hat p^{\lambda}_i(v_i, v_{-i}  ) < \hat p^{\lambda}_i(v'_i, v_{-i} )$.
Then, we have
\begin{align*}
u_i\left(\hat g^{\lambda}(v_i, v_{-i}); v_i\right)=v_i\big(\hat a^{\lambda}(v_i, v_{-i}  )\big)+\hat p^{\lambda}_i(v_i, v_{-i} ) &< v_i\big(\hat a^{\lambda}(v'_i, v_{-i}  )\big)+\hat p^{\lambda}_i(v'_i, v_{-i}  )\\
&=u_i\left(\hat g^{\lambda}(v'_i, v_{-i}); v_i\right),
\end{align*}
contradicting condition (iii) of Lemma \ref{lem_rev0}.
 
\medskip

{\bf Step 2.}
Let us show that for each $i \in I$ and each $v \in D^{\lambda} \setminus V_0$,
\[\hat p^{\lambda}_i(v) \geq \sum_{j \neq i}  v_j\big(\hat a^{\lambda}(v)\big) - \frac{n-1}{n}r .\]
For each $i \in I$ and each $v \in D^{\lambda} \setminus V_0$, let 
\[h_i(v) \equiv \hat p^{\lambda}_i(v)- \sum_{j \neq i} v_j\big(\hat a^{\lambda}(v)\big) + \frac{n-1}{n}r .\] 
It suffices to show that $h_i(v) \geq 0$ for all $i \in I$ and $v \in D^{\lambda} \setminus V_0$.

Take any $i \in I$.
Suppose, by contradiction, that there exists  $v \in D^{\lambda} \setminus V_0$ such that 
$h_i(v) < 0$.
Since $v \notin V_0$ and $\hat g^{\lambda}(v) \in  F(v)$, by \emph{efficiency} of $F$, we have $\hat a^{\lambda}(v) \neq \phi$. 
Let $\varepsilon>0$ be such that 
$\varepsilon = \min \Big\{-\frac{1}{2}h_i(v), \ \frac{1}{n}r \Big\}$. 
Define $v'_i \in V$ by
\[ v'_i(b)= 
\begin{cases}
0  &\text{if $b=\phi$}, \\
-\sum_{j \neq i} v_j\big(\hat a^{\lambda}(v)\big)+ \frac{n-1}{n}r +\varepsilon \ \ &\text{if $b=\hat a^{\lambda}(v)$}, \\
 -\sum_{j \neq i} v_j(b)  \ \ &\text{otherwise}.
\end{cases}
\]
Then, 
\[ v'_i(b)+\sum_{j \neq i} v_j(b)= 
\begin{cases}
 \frac{n-1}{n}r +\varepsilon \ \ &\text{if $b = \hat a^{\lambda}(v)$}, \\
0  \ \ &\text{otherwise}.
\end{cases}
\]
Note that by the definition of $\varepsilon$, $(v'_i, v_{-i}) \in D^{r}  \setminus V_0  \subset D^{\lambda}  \setminus V_0 $.
Moreover,  by \emph{efficiency} of $F$,
$\hat a^{\lambda}(v'_i, v_{-i})=\hat a^{\lambda}(v)$.
Therefore, by $v \in D^{\lambda} \setminus V_0$ and Step 1, $\hat p^{\lambda}_i(v'_i, v_{-i}  )=\hat p^{\lambda}_i(v_i, v_{-i}  )$, which implies that $h_i(v'_i, v_{-i}  )=h_i(v_i, v_{-i}  )<0$.
Then, by the definitions of $h_i$ and $v'_i$, 
\begin{align*}
v'_i\big(\hat a^{\lambda}(v'_i, v_{-i}  )\big)+\hat p^{\lambda}_i(v'_i, v_{-i}  )
&= v'_i\big(\hat a^{\lambda}(v'_i, v_{-i}  )\big)+\sum_{j \neq i} v_j\big(\hat a^{\lambda}(v'_i, v_{-i}  ) \big) - \frac{n-1}{n}r  + h_i(v'_i, v_{-i}  ) \\
&= h_i(v'_i, v_{-i}  )+\varepsilon \leq h_i(v'_i, v_{-i}  )-\frac{1}{2} h_i(v_i, v_{-i}  )
= \frac{1}{2} h_i(v_i, v_{-i}  )
<0.
\end{align*}
Since $\hat g^{\lambda}(v'_i, v_{-i}) =\left(\hat a^{\lambda}(v'_i, v_{-i}), \hat p^{\lambda}(v'_i, v_{-i} )\right) \in F(v'_i, v_{-i})$ by condition (ii) of Lemma \ref{lem_rev0}, the above inequality is a contradiction to \emph{individual rationality} of $F$.

\medskip

{\bf Step 3.}
Let us show that for each $i \in I$ and each $v \in D^{\lambda} \setminus V_0$,
\[ \hat p^{\lambda}_i(v)= \sum_{j \neq i}  v_j\big(\hat a^{\lambda}(v)\big)- \frac{n-1}{n}r . \]
Take any $i \in I$.
By Step 2, it suffices to show that for each $v \in D^{\lambda}  \setminus V_0$,
\[h_i(v) \equiv \hat p^{\lambda}_i(v)- \sum_{j \neq i} v_j\big(\hat a^{\lambda}(v)\big) + \frac{n-1}{n}r \leq 0.\] 
Suppose, by contradiction, that there exists  $v \in D^{\lambda} \setminus V_0$ such that $h_i(v)>0$.

Since $v \notin V_0$ and condition (ii) of Lemma \ref{lem_rev0} is satisfied, \emph{efficiency} of $F$ implies that $\hat a^{\lambda}(v) \neq \phi$.
Let $v'_i \in V$ be such that 
\[ v'_i(b)= 
\begin{cases}
0  \ \ &\text{if $b=\phi$}, \\
  -\sum_{j \neq i}  v_j\big(\hat a^{\lambda}(v)\big) +  r  \ \ &\text{if $b=\hat a^{\lambda}(v)$}, \\
-\sum_{j \neq i} v_j(b)  \ \ &\text{otherwise}.
\end{cases}
\]
Then, by the definition of $v'_i$, $(v'_i, v_{-i}) \in D^{r}  \setminus V_0  \subset D^{\lambda}  \setminus V_0 $.
Moreover, by \emph{efficiency} of $F$ and condition (ii) of Lemma \ref{lem_rev0}, we obtain
$\hat a^{\lambda}(v'_i, v_{-i})=\hat a^{\lambda}(v)$.
Therefore, by Step 1, $\hat p^{\lambda}_i(v'_i, v_{-i}  )=\hat p^{\lambda}_i(v_i, v_{-i}  )$, and thus
$h_i(v'_i, v_{-i}  )=h_i(v_i, v_{-i}  )>0$.
Note that by Step 2, $h_j(v'_i,v_{-i}) \geq 0$ holds for all $j \in I$, as $(v'_i, v_{-i}) \in D^{\lambda}  \setminus V_0 $.
Then, by the definitions of $h_i$ and $v'_i$, \begin{align*}
\sum_{j=1}^n \hat p^{\lambda}_j(v'_i, v_{-i})
&=(n-1)\Big( v'_i\big(\hat a^{\lambda}(v'_i, v_{-i}) \big)+\sum_{j \neq i} v_j\big(\hat a^{\lambda}(v'_i, v_{-i})\big) \Big)-  (n-1)r +\sum_{j = 1}^n h_j(v'_i, v_{-i}  ) \\
&\geq (n-1)\Big( v'_i\big(\hat a^{\lambda}(v'_i, v_{-i}) \big)+\sum_{j \neq i} v_j\big(\hat a^{\lambda}(v'_i, v_{-i})\big) \Big)-   (n-1) r + h_i(v'_i, v_{-i}  ) \\
&=(n-1)\Big( v'_i\big(\hat a^{\lambda}(v) \big)+\sum_{j \neq i} v_j\big(\hat a^{\lambda}(v)\big) \Big)- (n-1) r + h_i(v'_i, v_{-i}  ) \\
&=(n-1)r - (n-1) r + h_i(v'_i, v_{-i}  ) =h_i(v'_i, v_{-i}  )>0.
\end{align*}
Since $\left(\hat a^{\lambda}(v'_i, v_{-i}), \hat p^{\lambda}(v'_i, v_{-i} )\right) \in F(v'_i, v_{-i})$, this contradicts \emph{budget feasibility} of $F$.
\end{proof}

\subsection{Proof of Proposition~\ref{prop1} and statement~(i)
of Theorem~\ref{thm1}}
\label{app:interval-i}
In the proof of ``if" direction of statement (i) of Theorem \ref{thm1}, we need to define a family of mechanisms whose strategy spaces consist of pairs $(v_i, a_i) \in V \times A$ of valuations and alternatives to $\{D^{r}\}_{r \in \mathbb{R}_{++}}$-implement a quasi-equal-division rule $F^*$. This adjustment is purely technical and is required when $F^*$ selects multiple efficient alternatives for a single profile. The mechanisms in Proposition~\ref{prop1} correspond to the case
where $A^r(v)$, defined below, is a singleton whenever
$v\in D^r\setminus V_0$.
In this case, reports of alternatives can be omitted.
Hence, we restrict our presentation here to the proof of Theorem \ref{thm1}.

\medskip

{\bf \hypertarget{If_part}{``If" direction}.}
Let $F^*$ be a quasi-equal-division rule.
Obviously, $F^*$ satisfies \emph{efficiency} and \emph{individual rationality}.
Let us show that it also satisfies \emph{budget feasibility} and \emph{$\{D^{r}\}_{r \in \mathbb{R}_{++}}$-implementability}.

\medskip

{\bf \emph{Budget feasibility}.}
Take any $v \in V^n$ and any $(a, p) \in F^*(v)$.
By condition (iv) of the definition of quasi-equal-division rules, we have $0 \leq \sum_{j=1}^n u_j\left(a, p; v_j\right) \leq S(v)$.
Then, by \emph{efficiency} of $F^*$, 
\[\sum_{j=1}^n u_j\left(a, p; v_j\right) =\sum_{j=1}^n\left(v_j(a)+p_j\right) = S(v)+\sum_{j=1}^n p_j\leq S(v),\]
and thus $\sum_{j=1}^n p_j \leq 0$.

\bigskip

{\bf \emph{$\boldsymbol{\{D^r\}_{r \in \mathbb{R}_{++}}}$-implementability}.}
For each $v \in V^n$ and each $r \in \mathbb{R}_{++}$, define
\begin{align*}
&c^r(v)= n\cdot S(v)-(n-1) r,\\
&A^r(v)= 
\left\{a \in A: \exists p \in \mathbb{R}^n  \ \text{s.t.} \  (a, p) \in F^*(v) \ \text{and} \ \sum_{j=1}^n u_j\left(a, p; v_j\right)=c^r(v) \right\},
\end{align*}
where $A^r(v)$ is the set of alternatives prescribed by $F^*(v)$
at the total final utility level $c^r(v)$.
Note that if $v \in D^r \setminus V_0$, then $0 \leq c^r(v) \leq S(v)$, and thus 
by condition (v) of the definition of the quasi-equal-division rule, $A^r(v) \neq \emptyset$.
In addition, by \emph{efficiency} of $F^*$, all alternatives in $A^r(v)$ are efficient, i.e., $A^r(v) \subset \argmax_{b \in A} \sum_{j=1}^n v_j( b)$.
For each $r \in \mathbb{R}_{++}$ and each $v \in D^r \setminus V_0$, fix some $b^r(v) \in A^r(v)$.

For each $r \in \mathbb{R}_{++}$, define a mechanism $\left((V \times A)^n, g^{r}\right)$ such that each individual reports a pair consisting of a valuation function $v_i$ and an alternative $a_i$, and for each strategy profile $s=(v_j, a_j)_{j=1}^n \in (V \times A)^n,$ the following four conditions are satisfied:
\begin{itemize}
\item[(i)] If $S(v) <\frac{n-1}{n} r$, then $a^r\left(s\right)=\phi$,
\item[(ii)] If $\frac{n-1}{n} r \leq S(v) \leq r$,
 then $a^r\left(s\right)$ is determined as follows:
\begin{itemize}
\item If all individuals' reports on alternatives are identical, i.e., $a_i \equiv a$ for all $i \in I$, and $a \in A^r(v)$, then $a^r\left(s\right)=a$,
\item Otherwise, $a^r\left(s\right)=b^r(v)$,
\end{itemize}
\item[(iii)] If $S(v) >r $, then $a^r\left(s\right) \in \argmax_{b \in A} \sum_{j=1}^n v_j( b)$,
\item[(iv)] For each $i \in I$, 
$p^r_i\left( s \right)=\sum_{j \neq i} v_j\left( a^r\left( s \right) \right)+\frac{n-1}{n}r  \cdot \1_{\phi}\left( a^r\left( s \right)\right) -\frac{n-1}{n}r$. 
\end{itemize}

Let us show that $F^*$ is $\{D^r\}_{r \in \mathbb{R}_{++}}$-implemented by  $\{\left((V \times A)^n, g^{r}\right) \}_{r \in \mathbb{R}_{++}}$.

\medskip

{\bf Step 1.}
Let us show that for each $r \in \mathbb{R}_{++}$, each $i \in I$, and each $(v_i, a_i) \in V \times A$, reporting $(v_i, a_i)$ is a dominant strategy at $\bigl(\left((V \times A)^n, g^{r}\right), v_i\bigr)$.
In particular, once $i$ truthfully reports $v_i$, the choice of $a_i$ has no impact on $i$'s payoff.
Take any $r \in \mathbb{R}_{++}$, any $i \in I$, and any $s=(v_{j}, a_{j})_{j=1}^n \in (V \times A)^n$. 
It suffices to show that for each $s'_i=(v'_i, a'_i) \in V \times A$, 
\begin{align}
u_i\left(g^{r} \left(s_i, s_{-i}\right); v_i\right) \geq u_i\left(g^{r} \left(s'_i, s_{-i}\right); v_i\right). \label{v-de} 
\end{align}
As noted earlier, if $a \in A^r(v)$, then $a \in \argmax_{b \in A} \sum_{j=1}^n v_j( b)$.
Thus, by the definition of $g^{r}$, if $S(v) \geq \frac{n-1}{n} r$, then $a^r\left(s\right) \in \argmax_{b \in A} \sum_{j=1}^n v_j( b)$,
and if $S(v) < \frac{n-1}{n} r$, then $a^r\left(s\right)=\phi$.
Therefore, condition (iv) of the definition of $g^{r}$ yields that 
\begin{align*}
u_i\left(g^{r} \left(s_i, s_{-i}\right); v_i\right) =v_i\left(a^{r}(s) \right)+p^{r}_i(s) 
&=\max\Big\{0, \, S(v)-\frac{n-1}{n}r    \Big\},
\end{align*}
and
\begin{align*}
u_i\left(g^{r} \left(s'_i, s_{-i}\right); v_i\right)&=v_i\left(a^{r}(s'_i, s_{-i}) \right)+p^{r}_i(s'_i, s_{-i}) \\
&=\sum_{j =1}^n v_j\left(  a^{r}(s'_i, s_{-i}) \right)+\frac{n-1}{n}r  \cdot \1_{\phi}\left(  a^{r}(s'_i, s_{-i})\right) -\frac{n-1}{n}r\\
&\leq \max\Big\{0, \, \sum_{j=1}^n v_j\left( a^{r}(s'_i, s_{-i})\right)-\frac{n-1}{n}r    \Big\}.
\end{align*}
Since
$S(v) \geq \sum_{j=1}^n v_j\left( a^{r}(s'_i, s_{-i})\right)$,
we obtain equation \eqref{v-de}.

\medskip

{\bf Step 2.}
Let us show that for each $r \in \mathbb{R}_{++}$, each $i \in I$, and each $v_i \in V$, 
\begin{equation*}
DS_i\big(\left((V \times A)^n, g^{r}\right), v_i\big)=\left\{(v'_i, a_i) \in V \times A: v'_i=v_i\right\}.\label{dom-unique}
\end{equation*}
Take any $r \in \mathbb{R}_{++}$, any $i \in I$, and any $v_i \in V$.
By Step 1, it suffices to show that the left hand side of the above equation is a subset of the right hand side of that.
Suppose, by contradiction, that there exists $(v'_i, a_i) \in V \times A$ with $v'_i \neq v_i$ such that reporting $(v'_i, a_i)$ is a dominant strategy at $\big(\left((V \times A)^n, g^{r}\right), v_i\big)$.
Then, there exists $\hat a \in A$ such that $v'_i(\hat a) \neq v_i(\hat a)$.
Since $v'_i(\phi)=v_i(\phi)=0$, it must be $\hat a \neq \phi$.
Let $v_{-i} \in V^{n-1}$ be such that 
\begin{align*}
&\sum_{j \neq i} v_j(b) = 
\begin{cases}
0 \ \ &\text{if $b=\phi$}, \\
 - \frac{1}{2}  \bigl( v_i(\hat a) + v'_i(\hat a)\bigr) + \frac{n-1}{n}r  \ \ &\text{if $b=\hat a$}, \\
-  \max \{v_i(b), v'_i(b) \}  \ \ &\text{otherwise}. \\
\end{cases}
\end{align*}

For notational convenience, we write $s_i=(v_i, a_i)$, $s'_i=(v'_i, a_i)$, and $s_{-i}=(v_{j}, \hat a)_{j \neq i}$.
We consider two cases.

\medskip

{\bf Case 1.}  
Consider the case with $v_i(\hat a) > v'_i(\hat a)$.
By construction, $\hat a$ is the unique maximizer of
$\sum_{j=1}^n v_j(b)$ and
$S(v)>\frac{n-1}{n}r>S(v'_i,v_{-i})$.
Hence, by conditions (i)--(iii) of the definition of $g^r$,
$a^r(s_i,s_{-i})=\hat a$ and
$a^r(s'_i,s_{-i})=\phi$,
regardless of the reported alternatives.
Then,
\begin{align*}
u_i\left(g^{r} \left(s_i, s_{-i}\right); v_i\right)&=v_i\left(a^{r}(s_i, s_{-i})\right)+p_i^{r}(s_i, s_{-i}) \\
&=v_i(\hat a)+\sum_{j \neq i} v_j(\hat a) -\frac{n-1}{n}r 
=\frac{1}{2} v_i(\hat a) -\frac{1}{2} v'_i(\hat a) +\frac{n-1}{n}r  -\frac{n-1}{n}r \\
&=\frac{1}{2} \bigl(v_i(\hat a)- v'_i(\hat a) \bigr)\\
&> 0 =v_i\left(a^{r}(s'_i, s_{-i})\right)+p_i^{r}(s'_i, s_{-i})=u_i\left(g^{r} \left(s'_i, s_{-i}\right); v_i\right),
\end{align*}
contradicting the assumption that reporting $s'_i=(v'_i, a_i)$ is a dominant strategy at $\big(\left((V \times A)^n, g^{r}\right), v_i\big)$.

\medskip

{\bf Case 2.} 
Consider the case with $v_i(\hat a) < v'_i(\hat a)$.
By an argument similar to that in case 1, we have $a^{r}(s_i, s_{-i})= \phi$ and $a^{r}(s'_i, s_{-i})= \hat a$.
Then,
\begin{align*}
u_i\left(g^{r} \left(s'_i, s_{-i}\right); v_i\right)&=v_i\left(a^{r}(s'_i, s_{-i})\right)+p_i^{r}(s'_i, s_{-i}) \\
&=v_i(\hat a)+\sum_{j \neq i} v_j(\hat a) -\frac{n-1}{n}r 
=\frac{1}{2} v_i(\hat a) -\frac{1}{2} v'_i(\hat a) +\frac{n-1}{n}r  -\frac{n-1}{n}r \\
&=\frac{1}{2} \bigl(v_i(\hat a)- v'_i(\hat a) \bigr)\\
&< 0 =v_i\left(a^{r}(s_i, s_{-i})\right)+p_i^{r}(s_i, s_{-i})=u_i\left(g^{r} \left(s_i, s_{-i}\right); v_i\right),
\end{align*}
contradicting the assumption that reporting $s'_i=(v'_i, a_i)$ is a dominant strategy at $\big(\left((V \times A)^n, g^{r}\right), v_i\big)$.

\bigskip

{\bf Step 3.} 
Let us show that for each $r \in \mathbb{R}_{++}$ and each $v \in D^r$, $DE\big(\left((V \times A)^n, g^{r}\right), v\big) \neq \emptyset$ and $g^{r} \left(DE\big(\left((V \times A)^n, g^{r}\right), v\big)\right) \subset F^*(v)$.
By Step 2, the first part is clear.
Take any $r \in \mathbb{R}_{++}$, any $v \in D^r$, and any $s=(v_j, a_j)_{j=1}^n \in DE\big(\left((V \times A)^n, g^{r}\right), v\big)$.
Note that by Step 2, each $j \in I$ must report his/her true valuation $v_j$ as a dominant strategy.

First, consider the case with $S(v)=0$.
By the definition of $g^r$,  $a^r(s)=\phi$ and $p_i^r(s)=0$ for all $i \in I$.
Moreover, by condition (v) of the definition of $F^*$, there exists $(a, p) \in F^*(v)$ such that $u_i\left(a, p; v_i\right)=0$ for all $i \in I$.
Since $S(v)=0$, conditions (ii)-(iv) of the definition of $F^*$ imply that $a=\phi$ and $p_i=0$ for all $i \in I$.
Therefore, $a^r(s)=\phi=a$ and $p_i^r(s)=0=p_i$ for all $i \in I$, and thus
$g^r(s) \in F^*(v)$.

Second, consider the case with $v \in D^r \setminus V_0$, i.e., $\frac{n-1}{n} r \leq S(v) \leq r$.
By the definition of $g^r$, it follows that
\begin{align*}
&a^r(s) \in  A^r(v),\\
&p^r_i\left(s\right)=\sum_{j \neq i} v_j\left( a^r\left(s\right) \right)-\frac{n-1}{n}r  \ \text{for all $i \in I$.}
\end{align*} 
On the other hand, by the definition of $A^r(v)$, there exists $p \in \mathbb{R}^n$ such that $(a^r(s), p) \in F^*(v)$ and
$\sum_{j=1}^n u_j\left(a^r(s), p; v_j\right)=c^r(v)$.
This relation and condition (iii) of the definition of $F^*$ imply that
\begin{align}
u_i\left(a^r(s), p; v_i\right)=v_i\left(a^r(s)\right)+p_i=\frac{1}{n} c^r(v)=S(v)-\frac{n-1}{n} r  \ \ \text{for all $i \in I$,} \label{eq-util}
\end{align}
where the last equation follows from the definition of $c^r(v)$.
Note that by \emph{efficiency} of $F^*$,  we have
$S(v)= \sum_{j=1}^n v_j\left(a^r(s)\right)$.
Then, equation \eqref{eq-util} can be reformulated as 
\begin{align*}
p_i=\sum_{j \neq i} v_j\left(a^r(s)\right)-\frac{n-1}{n} r= p^r_i\left(s\right) \ \ \text{for all $i \in I$,}
\end{align*}
which in turn implies that $g^r(s)=\left(a^r(s), p^r(s)\right) = \left(a^r(s), p \right)\in F^*(v)$.

\bigskip

{\bf Step 4.} 
Let us show that for each $v \in V^n$, 
\[\bigcup_{\{r \in \mathbb{R}_{++}:  v \in D^{r}\}} g^{r} \left(DE\big(\left((V \times A)^n, g^{r}\right), v\big)\right)  = F^*(v).\]
Since the relation $\subset$ follows from Step 3, we show the opposite side.
Take any  $v \in V^n$ and $(a, p) \in F^*(v)$.
If $v \in V_0$, then by the definition of $F^*$, $a=\phi$ and $p_i=0$ for all $i \in I$, and by the definition of $g^r$ and Step 2, $(a, p) \in g^{r}\left(DE\big(\left((V \times A)^n, g^{r}\right), v\big)\right)$ for all $r \in \mathbb{R}_{++}$.

Suppose that $v \notin V_0$.
By the definition of $F^*$, $0 \leq \sum_{j=1}^n u_j\left(a, p; v_j\right) \leq S(v)$.
Let 
\[r=\left(\frac{n}{n-1}\right) \left(S(v)- \frac{1}{n}\sum_{j=1}^n u_j\left(a, p; v_j\right)\right)>0,\]
where the inequality follows from $S(v)>0$ as $v \notin V_0$.
Then,
\begin{align*}
\frac{n-1}{n}r
&=S(v)- \frac{1}{n}\sum_{j=1}^n u_j\left(a, p; v_j\right)\notag \\
&\leq S(v)=\left(\frac{n}{n-1}\right) \left(S(v)-\frac{1}{n} S(v) \right)\notag \\
&\leq \left(\frac{n}{n-1}\right) \left(S(v)- \frac{1}{n}\sum_{j=1}^n u_j\left(a, p; v_j\right) \right) = r,\label{v-in-Dr}
\end{align*}
and hence $v \in D^r$.
Moreover, by its own definition,
\[c^r(v)= n\cdot S(v)-(n-1)r=\sum_{j=1}^n u_j\left(a, p; v_j\right) .\]
Since $(a, p) \in F^*(v)$ holds, the above relation and the definition of $A^r(v)$ imply that $a \in A^r(v)$.
Therefore, by $\frac{n-1}{n}r \leq S(v) \leq r$ and the definition of $g^r$, we have $a^r\big((v_j, a)_{j=1}^n \big)=a$, where we will write $s \equiv (v_j, a)_{j=1}^n$.

Note that by $v \notin V_0$ and \emph{efficiency} of $F^*$, $a^r(s)=a \neq \phi$ and $S(v)=\sum_{j=1}^n v_j(a)$ must hold.
In addition, by condition (iii) of the definition of $F^*$, we have $u_i\left(a, p; v_i\right)=\frac{1}{n} \sum_{j=1}^n u_j\left(a, p; v_j\right) $ for all $i \in I$.
Therefore, for each $i \in I$, we obtain
\[u_i\left(a, p; v_i\right)=v_i(a)+p_i=\frac{1}{n} \sum_{j=1}^n u_j\left(a, p; v_j\right) = S(v)-\frac{n-1}{n}r,\]
where the last equation follows from the definition of $r$.
Thus, it follows that 
\begin{align*}
p_i&=S(v)-v_i(a)-\frac{n-1}{n}r
= \sum_{j=1}^n v_j(a)-v_i(a)-\frac{n-1}{n}r\\
&=\sum_{j \neq i} v_j(a)-\frac{n-1}{n}r
=\sum_{j \neq i} v_j\left(a^r(s)\right)-\frac{n-1}{n}r\\
&=\sum_{j \neq i} v_j\left(a^r(s)\right)+\frac{n-1}{n}r  \cdot \1_{\phi}\left(a^r(s)\right) -\frac{n-1}{n}r =p^r_i\left(s\right).
\end{align*}
By Step 2, $s=(v_j, a)_{j=1}^n \in DE\big(\left((V \times A)^n, g^{r}\right), v\big)$, and hence this and $v \in D^r$ together imply that
\begin{align*}
(a, p)=\left(a^r\left(s\right), p^r\left(s\right)\right) 
&\in g^r\left(DE\big(\left((V \times A)^n, g^{r}\right), v\big)\right) \\
&\subset  \bigcup_{\{r \in \mathbb{R}_{++}:  v \in D^{r}\}} g^{r} \left(DE\big(\left((V \times A)^n, g^{r}\right), v\big)\right).
\end{align*}
Thus, $F^*$ is $\{D^r\}_{r \in \mathbb{R}_{++}}$-implemented by the family of mechanisms $\{\left((V \times A)^n, g^{r}\right) \}_{r \in \mathbb{R}_{++}}$.
\hfill
$\square$

\bigskip

{\bf ``Only if" direction.}
\begin{claim}\label{charact1}
Let $F$ be a social choice rule that is efficient, individually rational, budget feasible, and $\{D^{r}\}_{r \in \mathbb{R}_{++}}$-implementable.
Fix any $ D^{r} \in \{D^{r}\}_{r \in \mathbb{R}_{++}}$.
Consider a mechanism $\left(V^n, \hat g^{r} \right)$ that satisfies conditions (ii) and (iii) in Lemma \ref{lem_rev0},
that is,
\begin{itemize}
\item for each $v \in D^{r}$, $\hat g^{r}(v) \in F(v)$; and
\item for each $i \in I$, each $v \in V^n$, and each $v'_i \in V$,
\[ u_i\left(\hat g^{r}(v_i, v_{-i}); v_i\right) \geq u_i\left(\hat g^{r}(v'_i, v_{-i}); v_i\right).\] 
\end{itemize}
Then, for each $i \in I$ and each $v \in D^{r}$, 
\[ \hat p^r_i(v)=\sum_{j \neq i} v_j\big(\hat a^r(v)\big)+\frac{n-1}{n}r   \cdot \1_{\phi}\big(\hat a^r(v)\big) -\frac{n-1}{n}r . \]
\end{claim}

\begin{proof}
{\bf Step 1.}
Let us show that  for each $i \in I$ and each $v \in V_0$,  $\hat p^r_i(v)=\sum_{j \neq i} v_j\big( \hat a^r(v)\big)$. 
First, we show that for each $i \in I$ and each $v \in V_0$, $\hat p^r_i(v) \geq \sum_{j \neq i} v_j\big( \hat a^r(v)\big)$. 
Suppose, by contradiction, that there exist  $i \in I$ and  $v \in V_0$ such that  $\hat p^r_i(v) < \sum_{j \neq i} v_j\big( \hat a^r(v)\big)$.
By $v \in V_0 \subset D^r$, we have $S(v)=0$ and $\hat g^{r}(v) =(\hat a^r(v), \hat p^r(v)) \in F(v)$.
Then, by $\hat p^r_i(v) < \sum_{j \neq i} v_j\big( \hat a^r(v)\big)$ and \emph{efficiency} of $F$,
\begin{align*}
v_i\big(\hat a^r(v)\big)+\hat p^r_i(v)&< v_i\big(\hat a^r(v)\big)+\sum_{j \neq i} v_j\big( \hat a^r(v)\big)=\sum_{j = 1}^n v_j\big( \hat a^r(v)\big) =S(v)=0,
\end{align*}
a contradiction to \emph{individual rationality} of $F$.

Next, let us show that for each $i \in I$ and each $v \in V_0$,  $\hat p^r_i(v) \leq \sum_{j \neq i} v_j\big( \hat a^r(v)\big)$. 
Suppose, by contradiction, that there exist $i \in I$ and  $v \in V_0$ such that $\hat p^r_i(v) >\sum_{j \neq i} v_j\big( \hat a^r(v)\big)$. 
Then, since $\hat p^r_j(v) \geq \sum_{k \neq j} v_k\big( \hat a^r(v)\big)$ holds for all $j \in I$,
\begin{align*}
\sum_{j=1}^n \hat p^r_j(v) 
> \sum_{j=1}^n \left\{ \sum_{k \neq j} v_k\big( \hat a^r(v)\big) \right\} 
=(n-1)\sum_{j=1}^n v_j\big( \hat a^r(v)\big)
=(n-1)S(v)  =0,
\end{align*}
contradicting \emph{budget feasibility} of $F$.

\bigskip

{\bf Step 2.}
Let us show that for each $v \in V_0$,  $\hat a^r(v)=\phi$. 
To show this, fix any $i \in I$ and any $v \in V_0$.
Suppose, by contradiction, that  $\hat a^r(v) \neq \phi$. 
Let $v'_i \in V$ be such that 
\[ v'_i(b)= 
\begin{cases}
0  \ \ &\text{if $b=\phi$}, \\
  -\sum_{j \neq i}  v_j\big(\hat a^r(v)\big) +  r  \ \ &\text{if $b=\hat a^r(v)$}, \\
-\sum_{j \neq i} v_j(b)  \ \ &\text{otherwise}.
\end{cases}
\]
Note that since $(v'_i, v_{-i}) \in D^r \setminus V_0$, we have
$\hat g^{r}(v'_i, v_{-i}) =(\hat a^r(v'_i, v_{-i}), \hat p^r(v'_i, v_{-i})) \in F(v'_i, v_{-i})$.
Then, by \emph{efficiency} of $F$, we have $\hat a^r(v'_i, v_{-i}) =\hat a^r(v)$.
In addition, by the assumption on $\hat g^{r}$, it follows that
\begin{align*}
v'_i\big(\hat a^r(v'_i, v_{-i}  )\big)+\hat p^r_i(v'_i, v_{-i})
\geq 
v'_i\big(\hat a^r(v_i, v_{-i}   )\big)+\hat p^r_i(v_i, v_{-i}).
\end{align*}
Then, by  $\hat a^r(v'_i, v_{-i}) =\hat a^r(v)$ and Step 1, 
\begin{equation} \label{peq}
\hat p^r_i(v'_i, v_{-i}) \geq \hat p^r_i(v_i, v_{-i})=\sum_{j \neq i} v_j\big( \hat a^r(v)\big).
\end{equation}
On the other hand, since  $(v'_i, v_{-i}) \in D^r \setminus V_0$, by invoking Lemma \ref{lem-imple}, we obtain
\[\hat p^r_i(v'_i, v_{-i})=\sum_{j \neq i} v_j\big( \hat a^r(v'_i, v_{-i})\big) -\frac{n-1}{n}r =\sum_{j \neq i} v_j\big( \hat a^r(v)\big)  -\frac{n-1}{n}r <\sum_{j \neq i} v_j\big( \hat a^r(v)\big), \]
which is a contradiction to equation \eqref{peq}.
Therefore, $\hat a^r(v) = \phi$ holds. 

\bigskip

{\bf Step 3.}
Let us show that for each $i \in I$ and each $v \in D^{r}$, 
\[ \hat p^r_i(v)=\sum_{j \neq i} v_j\big(\hat a^r(v)\big)+\frac{n-1}{n}r   \cdot \1_{\phi}\big(\hat a^r(v)\big) -\frac{n-1}{n}r . \]
Take any $i \in I$ and any  $v \in D^{r}$. 
Consider the case with  $v \in D^r \setminus V_0$.
Then, by $\hat g^{r}(v) =(\hat a^r(v), \hat p^r(v)) \in F(v)$ and \emph{efficiency} of $F$, $\hat a^r(v) \neq \phi$.
Therefore, by Lemma \ref{lem-imple},
\begin{align*}
\hat p^r_i(v)&=\sum_{j \neq i} v_j\bigl(\hat a^r(v)\bigr) -\frac{n-1}{n}r \\
&=\sum_{j \neq i} v_j\big(\hat a^r(v)\big)+\frac{n-1}{n}r  \cdot \1_{\phi}\big(\hat a^r(v)\big) -\frac{n-1}{n}r . 
\end{align*}
Next, consider the case with $v \in V_0$.
Then, by Step 2, $\hat a^r(v)=\phi$,
and thus by Step 1,
\[ \hat p^r_i(v)=\sum_{j \neq i} v_j\bigl(\hat a^r(v)\bigr)=\sum_{j \neq i} v_j\big(\hat a^r(v)\big)+\frac{n-1}{n}r  \cdot \1_{\phi}\big(\hat a^r(v)\big) -\frac{n-1}{n}r . \]
\end{proof}

Let us show that $F$ satisfies conditions (i)-(v) of the definition of the quasi-equal-division rules.

\medskip

\noindent
\emph{Conditions (i)--(iv):}
Take any $\hat v \in V^n$ and any $(\hat a, \hat p) \in F(\hat v)$.
Since $F$ is $\{D^{r}\}_{r \in \mathbb{R}_{++}}$-implementable, there exists a family of mechanisms $\left\{ \left(S_I^{r},g^{r} \right) \right\}_{r \in \mathbb{R}_{++}}$ that $\{D^{r}\}_{r \in \mathbb{R}_{++}}$-implements a social choice rule $F$. 
By the definition of $\{D^{r}\}_{r \in \mathbb{R}_{++}}$-implementation, there exists $r \in \mathbb{R}_{++}$ with $\hat v \in D^r$ such that $(\hat a, \hat p) \in g^{r} \left(DE(\left(S_I^{r},g^{r}\right), \hat v)\right)$.
Then, by Lemma \ref{lem_rev0}, there exists a mechanism $\left(V^n, \hat g^{r} \right)$ such that 
\begin{itemize}
\item[(i)] $\hat g^{r}(\hat v)=(\hat a, \hat p)$,
\item[(ii)] for each $v \in D^{r}$, $\hat g^{r}(v) \in F(v)$, 
\item[(iii)] for each $i \in I$, each $v \in V^n$, and each $v'_i \in V$,
\[ u_i\left(\hat g^{r}(v_i, v_{-i}); v_i\right) \geq u_i\left(\hat g^{r}(v'_i, v_{-i}); v_i\right).\] 
\end{itemize}

By \emph{efficiency} of $F$, we have $\hat a \in \argmax_{b \in A} \, \sum_{j=1}^n \hat v_j(b)$. Since $\hat v \in V^n$ and $(\hat a, \hat p) \in F(\hat v)$ were chosen arbitrarily, condition (i) of the definition of the quasi-equal-division rules holds.
Moreover, by the same argument as Step 2 of the proof of Claim \ref{charact1}, we can confirm that condition (ii) is also satisfied.

By  Claim \ref{charact1} and $\hat g^{r}(\hat v)=(\hat a, \hat p)$, for each $i \in I$, we have
\begin{align}
u_i\left(\hat a, \hat p; \hat v_i\right)=\hat v_i\left(\hat a\right)+\hat p_i&=\sum_{j =1}^n \hat v_j\big(\hat a\big)+\frac{n-1}{n}r   \cdot \1_{\phi}\big(\hat a\big) -\frac{n-1}{n}r \notag \\
&=S(\hat v) +\frac{n-1}{n}r   \cdot \1_{\phi}\big(\hat a\big) -\frac{n-1}{n}r, \label{eq1-quasi-claim2}
\end{align}
where the last equality follows from $\hat a \in \argmax_{b \in A} \, \sum_{j=1}^n \hat v_j(b)$.
Since the right hand side of \eqref{eq1-quasi-claim2} does not depend on $i$, condition (iii) is satisfied.

In addition, if $\hat v \in D^r \setminus V_0$, then $\hat a \neq \phi$ and $\frac{n-1}{n}r \leq S(\hat v) \leq r$.
If $\hat v\in V_0$, then $\hat a=\phi$ and $ S(\hat v)=\sum_{j=1}^n \hat v_j(\hat a)=0$.
Therefore, for each $i \in I$, 
\begin{align}
0  &\leq S(\hat v) +\frac{n-1}{n}r   \cdot \1_{\phi}\big(\hat a\big) -\frac{n-1}{n}r , \notag \\
& \leq S(\hat v) +\frac{n-1}{n}  S(\hat v) \cdot \1_{\phi}\big(\hat a\big) - \frac{n-1}{n} \cdot S(\hat v) \notag  \\
&=\frac{1}{n}S(\hat v).\label{eq2-quasi-claim2}
\end{align}
Combining equations \eqref{eq1-quasi-claim2} and \eqref{eq2-quasi-claim2}, we obtain $0 \leq  \sum_{j=1}^n u_j\left(\hat a, \hat p; \hat v_j\right) \leq S(\hat v)$, and thus condition (iv) is satisfied.

\medskip

\noindent
\emph{Condition (v):}
Consider any $\bar v \in V^n$ and any $0 \leq c \leq S(\bar v)$.
If $\bar v \in V_0$, let $r$ be any element of $\mathbb{R}_{++}$.
Otherwise let $r=\left(\frac{n}{n-1}\right) \left(S(\bar v)-\frac{1}{n}c\right)>0$.
If $\bar v\in V_0$, then clearly $\bar v\in D^r$.
If $\bar v \notin V_0$, then
\begin{align*}
\frac{n-1}{n}r
&=S(\bar v)- \frac{1}{n}c \leq S(\bar v)\\
&=\left(\frac{n}{n-1}\right) \left(S(\bar v)-\frac{1}{n} S(\bar v) \right)\notag \\
&\leq \left(\frac{n}{n-1}\right) \left(S(\bar v)- \frac{1}{n} c\right) = r,
\end{align*}
which implies that $\bar v \in D^r$.

Note that by Lemma \ref{lem_rev0}, there exists a mechanism $\left(V^n, \hat g^{r} \right)$ such that the two assumptions in Claim \ref{charact1} are met.
In particular, we have $\hat g^{r} (\bar v)=(\hat a^r(\bar v), \hat p^r(\bar v)) \in F(\bar v)$.
If $\bar v \in V_0$, then $0 \leq c \leq S(\bar v)=0$, and thus $c=0$.
Moreover, by the same argument as Steps 1 and 2 of Claim \ref{charact1}, we have $\hat a^r(\bar v)=\phi$ and $\hat p^r_i(\bar v)=0$ for all $i \in I$.
Therefore, $\sum_{j=1}^n u_j\left(\hat a^r(\bar v), \hat p^r(\bar v); \bar v_j\right)=0$,
 and thus condition (v) is satisfied.
Suppose $\bar v \notin V_0$. Then, by Claim \ref{charact1} and the definition of $r$, for each $i \in I$,
\begin{align*}
u_i\left(\hat a^r(\bar v), \hat p^r(\bar v);\bar v_i\right)&=\sum_{j =1}^n \bar v_j\big(\hat a^r(\bar v)\big)+\frac{n-1}{n}r   \cdot \1_{\phi}\big(\hat a^r(\bar v)\big) -\frac{n-1}{n}r \\
&=S(\bar v)  -\frac{n-1}{n}r =\frac{1}{n}c.
\end{align*}
The above equation implies $\sum_{j=1}^n u_j\left(\hat a^r(\bar v), \hat p^r(\bar v); \bar v_j\right)=c$, and hence condition (v) holds.
\hfill $\square$

\subsection{Proof of statement~(ii) of Theorem~\ref{thm1}}
\label{app:interval-ii}
Consider any domain cover $\{D^{\lambda}\}_{\lambda \in \Lambda}$ such that  $\{D^{r}\}_{r \in \mathbb{R}_{++}} \prec \{D^{\lambda}\}_{\lambda \in \Lambda}$.
Suppose, by contradiction, that there exists a social choice rule $F$ that is efficient, individually rational, budget feasible and $\{D^{\lambda}\}_{\lambda \in \Lambda}$-implementable.
Since $\{D^{r}\}_{r \in \mathbb{R}_{++}} \prec \{D^{\lambda}\}_{\lambda \in \Lambda}$, there exist $r \in \mathbb{R}_{++}$ and $\lambda \in \Lambda$ such that $D^r \subsetneq D^{\lambda}$, and $V_i (D^{\lambda})=V$ holds for each $i \in I$.
Then, by Lemma \ref{lem_rev0}, there exists a mechanism $\left(V^n, \hat g^{\lambda} \right)$ such that 
\begin{itemize}
\item for each $v \in D^{\lambda}$, $\hat g^{\lambda}(v) \in F(v)$; and
\item for each $i \in I$, each $v \in V^n$, and each $v'_i \in V$,
\[ u_i\left(\hat g^{\lambda}(v_i, v_{-i}); v_i\right) \geq u_i\left(\hat g^{\lambda}(v'_i, v_{-i}); v_i\right).\] 
\end{itemize}
Moreover, since $D^r \subsetneq D^{\lambda}$, there exists $v \in D^{\lambda} \setminus D^r$ such that one of the following two conditions holds:
\begin{itemize}
\item[(i)] $0< S(v) <\frac{n-1}{n}r$,
\item[(ii)]$r < S(v)$.
\end{itemize}

Suppose that  condition (i) holds.
Since $v \in  D^{\lambda}$, we have $\hat g^{\lambda}(v)=(\hat  a^{\lambda}(v), \hat p^{\lambda}(v)) \in F(v)$, and $v \notin D^r$ implies that $v \notin  V_0$.
Then, by Lemma \ref{lem-imple} and \emph{efficiency} of $F$,
\begin{align*}
   v_i\big(\hat a^{\lambda}(v)\big)+\hat p^{\lambda}_i(v)
   &= v_i\big(\hat a^{\lambda}(v)\big)+   \sum_{j \neq i}v_j\big(\hat a^{\lambda}(v)\big)-  \frac{n-1}{n}r 
   =  \sum_{j =1}^n v_j\big(\hat a^{\lambda}(v)\big) -  \frac{n-1}{n}r   \\
  &=  S(v)-  \frac{n-1}{n}r 
  <   \frac{n-1}{n}r - \frac{n-1}{n}r =0, 
\end{align*}
a contradiction to \emph{individual rationality} of $F$.

Suppose that condition (ii) holds.
Then, by $v  \in D^{\lambda} \setminus V_0$ and Lemma \ref{lem-imple},
\begin{align*}
\sum_{i=1}^n \hat p^{\lambda}_i(v)&=\sum_{i=1}^n \left(\sum_{j \neq i}v_j\big(\hat a^{\lambda}(v)\big)-  \frac{n-1}{n}r \right)\\
&=(n-1) \sum_{j=1}^n v_j\big(\hat a^{\lambda}(v)\big)  - (n-1)r 
=(n-1) S(v)  - (n-1) r \\
&>(n-1) r - (n-1) r=0,
\end{align*}
a contradiction to \emph{budget feasibility} of $F$.
\hfill
$\square$

\section{Proof of Proposition~\ref{prop2} and Theorem~\ref{thm2}}
\label{app:surplus}
The mechanisms in Proposition~\ref{prop2} correspond to the case
where $A(v)$, defined below, is a singleton in the construction
for the ``if'' direction of statement~(i) of Theorem~\ref{thm2}.
In this case, reports of alternatives can be omitted.
 Accordingly, we confine our presentation here to the proof of Theorem \ref{thm2}.

\subsection{Proof of statement~(i) of Theorem~\ref{thm2}}
\label{app:surplus-i}
{\bf \hypertarget{If_part2}{``If" direction}.}
Let $F^{**}$ be an equal-division rule.
Obviously, $F^{**}$ is \emph{efficient}, \emph{individually rational}, and \emph{budget balanced}.
Let us confirm that $F^{**}$ is $\{C^{r}\}_{r \in \mathbb{R}_{++}}$-implementable.
For each $v \in V^n$, let $A(v)=\left\{a \in A: \exists p \in \mathbb{R}^n \ \text{s.t.} \ (a, p) \in F^{**}(v) \right\}$.
By the definition of the equal-division rule, $A(v) \neq \emptyset$.
For each $v \in V^n$, fix some $b(v) \in A(v)$.

For each $r \in \mathbb{R}_{++}$, define a mechanism $\left((V \times A)^n, g^{r}\right)$ by the following four conditions: for each strategy profile $s=(v_j, a_j)_{j \in I} \in (V \times A)^n$, 
\begin{itemize}
\item[(i)] if $S(v) <\frac{n-1}{n} r$, then $a^r\left(s\right)=\phi$,
\item[(ii)] if $S(v) = r$, i.e., $v \in C^r \setminus V_0$, then $a^r\left(s\right)$ is determined as follows:
\begin{itemize}
\item if all individuals' reports on alternatives are identical, i.e., $a_i \equiv a$ for all $i \in I$, and $a \in A(v)$, then $a^r\left(s\right)=a$,
\item otherwise, $a^r\left(s\right)=b(v)$,
\end{itemize}
\item[(iii)] if $\frac{n-1}{n} r \leq S(v) <r$ or $r<S(v) $, then
\[a^r\left(s\right) \in \argmax_{b \in A} \sum_{j=1}^n v_j( b),\]
\item[(iv)] for each $i \in I$, 
\[p^r_i\left(s\right)=\sum_{j \neq i} v_j\left( a^r\left(s\right) \right)+\frac{n-1}{n}r  \cdot \1_{\phi}\left( a^r\left(s\right) \right) -\frac{n-1}{n}r .\] 
\end{itemize}

Let us show that $F^{**}$ is $\{C^{r}\}_{r \in \mathbb{R}_{++}}$-implemented by $\left\{\left((V \times A)^n, g^{r}\right)\right\}_{r \in \mathbb{R}_{++}}$.
Note that for each $r \in \mathbb{R}_{++}$, the mechanism $\left((V \times A)^n, g^{r}\right)$ takes a quite similar form to that defined in the proof of Theorem \ref{thm1} (``if" direction of statement (i) in Appendix~\ref{app:interval-i}).
Then, by the same argument as before, for each $r \in \mathbb{R}_{++}$, each $i \in I$, and each $v_i \in V$, we have
\begin{equation}
DS_i\big(\left((V \times A)^n, g^{r}\right), v_i\big)=\left\{(v'_i, a_i) \in V \times A: v'_i=v_i\right\}.\label{dom-unique2}
\end{equation}

\medskip

{\bf Step 1.} 
Let us show that for each $r \in \mathbb{R}_{++}$ and each $v \in C^r$, $DE\big(\left((V \times A)^n, g^{r}\right), v\big) \neq \emptyset$ and $g^{r} \left(DE\big(\left((V \times A)^n, g^{r}\right), v\big)\right) \subset F^{**}(v)$.
By equation \eqref{dom-unique2}, the first part is obvious.
Take any $r \in \mathbb{R}_{++}$, any $v \in C^r$, and any $s=(v_j, a_j)_{j \in I} \in DE\big(\left((V \times A)^n, g^{r}\right), v\big)$.
First, consider the case with $v \in V_0$.
Then, by the definition of $g^r$, $a^r\left(s\right)=\phi$ and $p_i^r\left(s\right)=0$ for all $i \in I$.
Moreover, by $S(v)=0$ and the definition of $F^{**}$, for any $(a, p) \in F^{**}(v)$, we have $a=\phi$ and $u_i\left(a, p; v_i\right)=0$ for all $i \in I$, which in turn implies that $p_i=0$ for all $i \in I$.
Therefore, we obtain $\left(a^r\left(s \right), p^r\left(s\right)\right)=(a, p) \in F^{**}(v)$.

Second, consider the case with $v \in C^r \setminus V_0$.
By the definition of $g^r$, we have
\begin{align*}
&a^r\left(s \right) \in A(v),\\
&p^r_i\left(s\right)=\sum_{j \neq i} v_j\left( a^r\left(s\right) \right)-\frac{n-1}{n}r  \ \text{for all $i \in I$.}
\end{align*} 
Moreover, by the definitions of $A(v)$ and $F^{**}$, there exists $p \in \mathbb{R}^n$ such that $(a^r\left(s \right), p) \in F^{**}(v)$  and
\begin{align}
v_i\left(a^r\left(s\right)\right)+p_i&=\frac{1}{n}S(v)  \ \ \text{$\forall i \in I$.} \label{eq-util2}
\end{align}
Since $S(v)=r$ and $a^r(s)$ is an efficient alternative, we also have
\begin{align}
\frac{1}{n}S(v)&=S(v)-\frac{n-1}{n} S(v)=S(v)-\frac{n-1}{n} r \notag\\
&=\sum_{j=1}^n v_j\left(a^r(s)\right)-\frac{n-1}{n} r. \label{eq-util3}
\end{align}
Equations \eqref{eq-util2} and \eqref{eq-util3} together imply that
\begin{align*}
p_i=\sum_{j \neq i} v_j\left(a^r(s)\right)-\frac{n-1}{n} r= p^r_i\left(s\right) \ \ \text{$\forall i \in I$.}
\end{align*}
Therefore, we obtain $\left(a^r(s), p^r(s)\right) =\left(a^r(s), p\right) \in F^{**}(v)$.

\bigskip

{\bf Step 2.} 
Let us show that for each $v \in V^n$, 
\begin{equation}
\bigcup_{\{r \in \mathbb{R}_{++}:  v \in C^{r}\}} g^{r} \left(DE\big(\left((V \times A)^n, g^{r}\right), v\big)\right)  = F^{**}(v).\label{thm2-imple2}
\end{equation}
Since the relation $\subset$ follows from Step 1, we show the opposite direction. 
Take any  $v \in V^n$ and $(a, p) \in F^{**}(v)$, and
define $s \equiv (v_j, a)_{j=1}^n \in (V \times A)^n$.
Note that by equation \eqref{dom-unique2}, we have $s \in  DE\big(\left((V \times A)^n, g^{r}\right), v\big)$.

First, suppose that $v \in V_0$.
Then, by the definition of $F^{**}$, $a=\phi$ and $p_i=0$ for all $i \in I$.
In addition, by the definition of $g^r$, we have $a^r(s)=\phi=a$ and $p_i^r(s)=0=p_i$ for all $i \in I$. Since $s \in  DE\big(\left((V \times A)^n, g^{r}\right), v\big)$,
it follows that $(a, p)=(a^r(s), p^r(s)) \in g^r\left(DE\big(\left((V \times A)^n, g^{r}\right), v\big)\right)$ for all $r \in \mathbb{R}_{++}$.

Second, suppose that $v \notin V_0$.
By the definition of $F^{**}$, we have $u_i\left(a, p; v_i\right)=\frac{1}{n}S(v)$ for all $i \in I$.
Let $r \equiv S(v)>0$. 
Then, we have $v \in C^r $,
and by the definition of $A(v)$, $a \in A(v)$ holds.
Thus, by the definition of $g^r$, we obtain $a^r(s)=a$.

Since $(a, p) \in F^{**}(v)$, \emph{efficiency} of $F^{**}$ implies that $a \neq \phi$ and $r=S(v)=\sum_{j=1}^n v_j(a)$.
Thus, for each $i \in I$, we have
\[v_i(a)+p_i=\frac{1}{n}S(v)= \sum_{j=1}^n v_j(a)-\frac{n-1}{n}r,\]
or equivalently,
\begin{align*}
p_i&= \sum_{j \neq i} v_j(a)-\frac{n-1}{n}r=\sum_{j \neq i} v_j\left(a^r(s)\right)-\frac{n-1}{n}r\\
&=\sum_{j \neq i} v_j\left( a^r\left(s\right) \right)+\frac{n-1}{n}r  \cdot \1_{\phi}\left( a^r\left(s\right) \right) -\frac{n-1}{n}r =p^r_i\left(s\right),
\end{align*}
where the second and third equations follow from $a^r(s)=a$ and $a^r(s) \neq \phi$, respectively.

Then, by $s=(v_j, a)_{j \in I} \in DE\big(\left((V \times A)^n, g^{r}\right), v\big)$, we have 
\begin{align*}
(a, p)=\left(a^r\left(s\right), p^r\left(s\right)\right) 
\in g^r\left(DE\big(\left((V \times A)^n, g^{r}\right), v\big)\right).
\end{align*}
Since $v \in C^r$, the above relation yields equation \eqref{thm2-imple2}, and thus $F^{**}$ is $\{C^r\}_{r \in \mathbb{R}_{++}}$-implemented by the family of mechanisms $\left\{\left((V \times A)^n, g^{r}\right)\right\}_{r \in \mathbb{R}_{++}}$.
\hfill
$\square$

\bigskip

{\bf ``Only if" direction.}

\begin{claim}\label{charact-balance1}
Assume that $|A|\geq3$.
Consider any social choice rule $F$ that is efficient, individually rational, budget balanced, and $\{C^{r}\}_{r \in \mathbb{R}_{++}}$-implementable.
Fix any $r \in \mathbb{R}_{++}$.
Let $\left(V^n, \hat g^{r} \right)$ be any mechanism that satisfies conditions (ii) and (iii) in Lemma \ref{lem_rev0},\footnote{Note that $V_i (C^r)=V$ for all $i \in I$, and thus each individual's strategy space becomes $V$.}
 that is,
\begin{itemize}
\item for each $v \in C^r$, $\hat g^{r}(v) \in F(v)$, 
\item for each $i \in I$, each $v \in V^n$, and each $v'_i \in V$,
$ u_i\left(\hat g^{r}(v_i, v_{-i}); v_i\right) \geq u_i\left(\hat g^{r}(v'_i, v_{-i}); v_i\right).$
\end{itemize}
Then, for each $i \in I$ and each $v \in V^n$, we have
\begin{align}
&\hat a^r(v) \in \argmax_{b \in A} \ \left\{\sum_{j=1}^n v_j(b) + \frac{n-1}{n}r  \cdot \1_{\phi}(b)\right\}, \label{balance-mech1} \\
&\hat p_i^r(v)= \sum_{j\neq i} v_j(\hat a^r(v) ) + \frac{n-1}{n}r  \cdot \1_{\phi}(\hat a^r(v) )-\frac{n-1}{n}r .\label{balance-mech2}
\end{align}
\end{claim}

\medskip

\begin{proof}
{\bf Step 1.}
Let us show that there exist $\alpha \in \mathbb{R}$ and $\left\{h_i: V^{n-1} \to \mathbb{R}\right\}_{i \in I}$ such that for each $i \in I$ and each $v \in V^n$,
\begin{align} 
&\hat a^r(v) \in \argmax_{b \in A} \left\{\sum_{j=1}^n v_j(b) + \alpha \cdot \1_{\phi}(b)\right\},\label{a^r} \\
&\hat p_i^r(v)= \sum_{j\neq i} v_j( \hat a^r(v)) + \alpha \cdot \1_{\phi}( \hat a^r(v))+ h_i(v_{-i}). 
\label{p^r}
\end{align}

Since the second half of the assumption means that the outcome function $\hat g^r$ satisfies \emph{strategy-proofness}, we can apply Roberts' (1979) Theorem (in the form of \citealt{CarbajalMcLennanTourky2013}; Theorem 1) to show that one of the following three cases occurs\footnote{Case 1 follows from the fact that any \emph{lexicographic affine maximizer} is also an \emph{affine maximizer}.}:

\begin{quote}
{\bf Case 1.}  there exist $\sigma \in \mathbb{R}^n_+$ with $\sigma_1+ \dots +\sigma_n=1$ and $w: A   \to \mathbb{R}$ such that for each $v \in V^n$,
\[\hat a^r(v) \in \argmax_{b \in A} \sum_{j=1}^n \sigma_j v_j(b)+w(b),\]

\medskip

{\bf Case 2.} there exist distinct $a, b \in A  $ with $\hat a^r(V^n) = \{a, b\}$ such that for each $i \in I$ and each $v_{-i} \in V^{n-1}$, there exists $\delta_i(v_{-i}) \in \mathbb{R} \cup \{-\infty, +\infty\}$ for which for each  $v_i \in V$,
\begin{align*}
&v_i(a)-v_i(b)> \delta_i(v_{-i}) \Longrightarrow \hat a^r(v)=a, \ \ \text{and} \\
&v_i(a)-v_i(b)< \delta_i(v_{-i}) \Longrightarrow \hat a^r(v)=b.
\end{align*}

\medskip

{\bf Case 3.} The decision function $\hat a^r$ is a constant function.
\end{quote}

By the first half of the assumption and \emph{efficiency} of $F$, $\hat a^r$ must choose an efficient alternative on $C^r$, implying that cases 2 and 3 lead to contradictions as $|A| \geq 3$.

Let us consider case 1.
Suppose, by contradiction, that $\sigma_i < \sigma_j$ for some $i, j \in I$.
Then, since $A$ is finite, there exists $\beta \in \mathbb{R}_{++}$ such that for each $b \in A \setminus \{\phi\}$,
\begin{equation}\label{sigma}
\sigma_j r+(\sigma_i -\sigma_j) \beta + w(b)< w(\phi).
\end{equation}
Let $v_i, v_j \in V$ be such that for each $b \in A \setminus \{\phi\}$, $v_i(b)=\beta$ and $v_j(b)=r-\beta$.
For each $k \in I \setminus \{i,j\}$, let $v_k \in V$ be such that $v_k(b)=0$ for all $b \in A$.
Then, $v \in C^r \setminus V_0$ holds, and by \emph{efficiency} of $F$, $\hat a^r(v) \neq \phi$.
On the other hand, by equation \eqref{sigma}, for each $b \in A \setminus \{\phi\}$,
\[  \sum_{j=1}^n \sigma_j v_j(b)+w(b)=\sigma_j r+(\sigma_i -\sigma_j) \beta + w(b) <w(\phi).  \]
Since case 1 occurs, the above relation implies that $\hat a^r(v)=\phi$, which is a contradiction.
Therefore,  we have $\sigma_i = \sigma_j$ for all $i, j \in I$. 

Next, suppose, by contradiction, there exist $a, b \in A \setminus \{\phi\}$ such that $w(a) < w(b)$.
Let $\sigma\equiv \sigma_1=\dots=\sigma_n$.
Since $\hat a^r$ chooses an efficient alternative on $C^r$, we have $\sigma>0$.
Then, we can easily construct $v \in C^r$ that satisfies the following three conditions:
(i) $ \sum_{j=1}^n v_j(a)=r$, (ii) $\sum_{j=1}^n v_j(c)<r$ for all $c \in A \setminus \{a\}$, and (iii) $\sigma\sum_{j=1}^n v_j(a) +w(a) < \sigma \sum_{j=1}^n v_j(b) +w(b)$.
However, this means that $\hat a^r(v) \neq a$, contradicting the fact that $\hat a^r$ chooses an efficient alternative on $C^r$.
Therefore, we have $w(a) = w(b)$ for all $a, b \in A \setminus \{\phi\}$,
and thus equation \eqref{a^r} holds. 
Equation \eqref{p^r} follows from the famous payoff equivalence theorem (see \citealt{ChungOlszewski2007}; Theorem 4).

\medskip

{\bf Step 2.}
Let us show that for each $i \in I$ and each $v_{-i} \in V^{n-1}$, $h_i(v_{-i})=-\alpha$.
First, we show that $h_i(v_{-i}) \geq -\alpha$.
Take any $i \in I$ and any $v_{-i} \in V^{n-1}$.
Define $v_i \in V$ in such a way that $S(v)=0$ and $\argmax_{b \in A} \left\{\sum_{j=1}^n v_j(b) + \alpha \cdot \1_{\phi}(b)\right\}=\{\phi\}$.
Then, $v \in V_0 \subset C^r$.
In addition, by Step 1, we have $\hat a^r(v)=\phi$ and $\hat p_i^r(v)=\alpha+h_i(v_{-i})$, and thus
$u_i\left(\hat g^r(v); v_i\right)=\alpha+h_i(v_{-i})$.
Since $\hat g^{r}(v) \in F(v)$ follows from the first half of our assumption,
\emph{individual rationality} of $F$ implies that $h_i(v_{-i}) \geq -\alpha$.

Second, we show that for each $i \in I$ and each $v_{-i} \in V^{n-1}$, $h_i(v_{-i}) \leq -\alpha$.
Suppose, by contradiction, that there exist $i \in I$ and $v_{-i} \in V^{n-1}$ such that $h_i(v_{-i}) > -\alpha$.
Let $v_i \in V$ be such that $S(v)=0$ and $\argmax_{b \in A} \left\{\sum_{j=1}^n v_j(b) + \alpha \cdot \1_{\phi}(b)\right\}=\{\phi\}$.
Then, $\hat a^r(v)=\phi$, $\hat p_i^r(v)=\alpha+h_i(v_{-i})>0$, and $\hat p_j^r(v) =\alpha+h_j(v_{-j}) \geq 0$ for all $j \neq i$, which imply that $\sum_{j=1}^n \hat p_j^r(v)>0$.
However, since $\hat g^{r}(v) \in F(v)$, this is a contradiction to \emph{budget balancedness} of $F$.

\medskip

{\bf Step 3.}
Let us show that $\alpha=\frac{n-1}{n}r $.
Consider any $v \in C^r \setminus V_0$, i.e., $S(v)=r$.
Then, by the first half of our assumption, $\hat g^r(v)  \in F(v)$ holds.
In addition, by \emph{efficiency} of $F$, we have $\hat a^r(v) \neq \phi$. 
Therefore, by \emph{budget balancedness} of $F$ and Steps 1 and 2, 
\begin{align*}
0=\sum_{j=1}^n \hat p^r_j(v)= (n-1)S(v) + \sum_{j=1}^n h_j(v_{-j})
= (n-1)r - n \alpha,   
\end{align*} 
or equivalently, $\alpha=\frac{n-1}{n}r $.
Hence, by Steps 1--3, equations \eqref{balance-mech1} and \eqref{balance-mech2} hold as desired.
\end{proof}

\medskip

To complete the proof, 
consider any social choice rule $F$ that is efficient, individually rational, budget balanced, and $\{C^{r}\}_{r \in \mathbb{R}_{++}}$-implementable.
Take any  $\hat v \in V^n$ and any $\hat x=(\hat a, \hat p) \in F(\hat v)$.
By condition (ii) of Definition~\ref{def-D}, there exists
$r\in\mathbb{R}_{++}$ such that $\hat v\in C^r$ and
$\hat x$ is a dominant-strategy outcome of the corresponding
implementing mechanism at $\hat v$.
By Lemma~\ref{lem_rev0}, there exists a mechanism
$\left(V^n,\hat g^r\right)$ such that
\begin{itemize}
\item[(i)] $\hat g^{r}(\hat v)=(\hat a, \hat p)$,
\item[(ii)] for each $v \in C^r$, $\hat g^{r}(v) \in F(v)$, 
\item[(iii)] for each $i \in I$, each $v \in V^n$, and each $v'_i \in V$,
$ u_i\left(\hat g^{r}(v_i, v_{-i}); v_i\right) \geq u_i\left(\hat g^{r}(v'_i, v_{-i}); v_i\right).$
\end{itemize}

By \emph{efficiency} of $F$ and $(\hat a, \hat p) \in F(\hat v)$, we have $\hat a \in \argmax_{b \in A} \sum_{j=1}^n \hat v_j(b)$, meaning that condition (i) of the definition of the equal-division rules is satisfied for $\hat a$.
In addition, by Claim \ref{charact-balance1}, if $S(\hat v)=0$, then $\hat a^{r}(\hat v)=\hat a=\phi$, and thus condition (ii) holds. 

Let us show that condition (iii) of the definition of the equal-division rules is satisfied.
Consider any $i \in I$.
If $\hat v \in V_0$, then by Claim \ref{charact-balance1},
\begin{align*}
u_i(\hat a, \hat p; \hat v_i)&=u_i(\hat a^{r}(\hat v), \hat p^{r}(\hat v); \hat v_i)=\hat v_i(\hat a^{r}(\hat v))+\hat p_i^{r}(\hat v)\\
&=\max_{b \in A} \, \left(\sum_{j=1}^n \hat v_j(b) +\frac{n-1}{n}r   \cdot \1_{\phi}(b) \right) -\frac{n-1}{n}r \\
& =\sum_{j=1}^n \hat v_j(\phi) +\frac{n-1}{n}r -\frac{n-1}{n}r =0=\frac{1}{n}S(\hat v). 
\end{align*}
On the other hand, if $\hat v \in C^r \setminus V_0$, then by Claim \ref{charact-balance1} and $S(\hat v)=r$, we have
\begin{align*}
u_i(\hat a, \hat p; \hat v_i)&=u_i(\hat a^{r}(\hat v), \hat p^{r}(\hat v); \hat v_i)=\hat v_i(\hat a^{r}(\hat v))+\hat p_i^{r}(\hat v)\\
&=\max_{b \in A} \,  \left(\sum_{j=1}^n \hat v_j(b) +\frac{n-1}{n}r   \cdot \1_{\phi}(b) \right) -\frac{n-1}{n}r \\
&=S(\hat v) -\frac{n-1}{n}r \\
&=S(\hat v) -\frac{n-1}{n}S(\hat v) 
=\frac{1}{n}S(\hat v). 
\end{align*}
Therefore, condition (iii) of the definition of the equal-division rules holds.
\hfill$\square$

\subsection{Proof of statement~(ii) of Theorem~\ref{thm2}}
\label{app:surplus-ii}
Consider any domain cover $\{D^\lambda\}_{\lambda\in\Lambda}$
such that $\{C^r\}_{r\in\mathbb{R}_{++}}
\prec
\{D^\lambda\}_{\lambda\in\Lambda}$.
Suppose, by contradiction, that there exists a social choice rule $F$
that is efficient, individually rational, budget balanced, and $\{D^{\lambda}\}_{\lambda \in \Lambda}$-implementable.
Since $\{C^r\}_{r\in\mathbb{R}_{++}}
\prec
\{D^\lambda\}_{\lambda\in\Lambda}$, there exist $r\in\mathbb{R}_{++}$ and $\lambda\in\Lambda$
such that $C^r\subsetneq D^\lambda$.
Then, by Lemma~\ref{lem_rev0} and the same argument as in
the proof of Claim~\ref{charact-balance1}, we can confirm the existence of a mechanism $\left(V^n, \hat g^r \right)$ such that 
\begin{itemize}
\item[(i)] for each $v \in D^{\lambda}$, $\hat g^{r}(v) \in F(v)$, 
\item[(ii)] for each $v \in V^n$, $\hat a^r(v) \in \argmax_{b \in A} \ \left\{\sum_{j=1}^n v_j(b) + \frac{n-1}{n}r  \cdot \1_{\phi}(b)\right\},$
\item[(iii)] for each $v \in V^n$ and each $i \in I$, 
$\hat p_i^r(v)= \sum_{j\neq i} v_j(\hat a^r(v) ) + \frac{n-1}{n}r  \cdot \1_{\phi}(\hat a^r(v) )-\frac{n-1}{n}r $.
\end{itemize}

Consider any $v \in D^{\lambda} \setminus C^{r}$.
Define $r' \in \mathbb{R}$ by $r' \equiv S(v)$.
Note that since $v \notin C^{r}$, we have $v \notin V_0$, and thus $r' \neq 0$ and $r' \neq r$.
Moreover, since $v \in  D^{\lambda}$, $\hat g^{r}(v)=(\hat a^r(v), \hat p^r(v)) \in F(v)$.
Therefore, by \emph{efficiency} of $F$, $\hat a^r(v) \neq \phi$ holds.
Then, by conditions (ii) and (iii) above, it follows that
\begin{align*}
\sum_{i=1}^n \hat p_i^r(v) 
&=\sum_{i=1}^n \left\{\sum_{j \neq i} v_j\left(\hat a^r(v)\right) -\frac{n-1}{n}r \right\}
=(n-1) \sum_{j=1}^n v_j\left(\hat a^r(v)\right)-r(n-1)\\
&=(n-1) S(v)-r(n-1)=r'(n-1)-r(n-1) \neq 0,
\end{align*}
contradicting \emph{budget balancedness} of $F$, as $(\hat a^r(v), \hat p^r(v)) \in F(v)$.
\hfill
$\square$

\section{Proof of Corollaries~\ref{coro-p1} and \ref{coro-pp4}}
\label{app:verification}
In the proofs of statement~(i) of Theorems~\ref{thm1}
and \ref{thm2}
(Appendices~\ref{app:interval-i} and \ref{app:surplus-i}),
we showed that for each quasi-equal-division rule $F^*$
(equal-division rule $F^{**}$), there exists a family of mechanisms
$\{\left((V\times A)^n,g^r\right)\}_{r\in\mathbb{R}_{++}}$
that $\{D^r\}_{r\in\mathbb{R}_{++}}$-implements $F^*$
($\{C^r\}_{r\in\mathbb{R}_{++}}$-implements $F^{**}$).
We confirmed that 
for each $r \in \mathbb{R}_{++}$, each $v \in V^n$, and each $i \in I$, \begin{align*}
&DS_i\big(\left((V \times A)^n, g^{r}\right), v_i\big)=\left\{(\bar v_i, a_i) \in V \times A: \bar v_i=v_i\right\}.
\end{align*}
Therefore, for each $r \in \mathbb{R}_{++}$, each $v \in D^{r}$ ($ \in C^r$) and each $v' \in V^n \setminus D^{r}$ ($ \in V^n \setminus C^r$), $DE\big(\left((V \times A)^n, g^{r}\right), v\big)$ and $DE\big(\left((V \times A)^n, g^{r}\right), v'\big)$ are nonempty and disjoint.
Hence, these families of mechanisms verify label validity under
$\{D^r\}_{r\in\mathbb{R}_{++}}$
and $\{C^r\}_{r\in\mathbb{R}_{++}}$, respectively.
\hfill
$\square$

\section*{Acknowledgements}

This paper was previously titled
``Mechanism Design Incorporating Executor's Domain Selection.''
The author thanks Toyotaka Sakai, Dov Samet, Shuhei Otani,
Yoko Kawada, Noriaki Okamoto, Takako Fujiwara-Greve,
Toru Hokari, Shinsuke Nakamura, and participants at the
14th Meeting of the Society for Social Choice and Welfare.
This research is financially supported by JSPS KAKENHI
Grant Number JP24K16334.

\section*{Declaration of generative AI and AI-assisted technologies
in the manuscript preparation process}

The author used ChatGPT (OpenAI) for language editing and for assistance in surveying the related literature. All results and substantive content are
the author's own. The author reviewed and edited the content
as needed and takes full responsibility for the manuscript.

\end{document}